\documentclass[11pt,prd,aps,amsfonts,eqsecnum,superscriptaddress,nofootinbib,longbibliography,notitlepage]{revtex4-1}

\usepackage{lmodern}    
\usepackage{microtype}  
\usepackage[english]{babel}
\usepackage{setspace}

\usepackage{fancyhdr}

\usepackage{caption}
\usepackage{graphicx}
\usepackage{xcolor}
\usepackage{subcaption}
\usepackage{rotating}
\usepackage{amsmath, amssymb, graphics, amsthm, isomath}
\usepackage{physics}
\usepackage{verbatim}
\usepackage{tikz}
\usepackage{tikz-cd}
\usepackage{circuitikz}
\usetikzlibrary{patterns}
\usetikzlibrary{through,hobby}
\usetikzlibrary{decorations.markings}
\usetikzlibrary{fadings,shadings}
\usetikzlibrary{plotmarks}
\usetikzlibrary{positioning}
\usetikzlibrary{decorations,arrows}
\usetikzlibrary{decorations.pathreplacing}
\usetikzlibrary{chains, scopes, positioning, backgrounds, shapes, fit, shadows, calc, arrows.meta, decorations.pathreplacing}
\usepackage{blkarray}
\usepackage{svg}
\usepackage{mathrsfs}
\usepackage{algorithm}
\usepackage{algpseudocode}

\usepackage{xcolor}
\definecolor{mycitecolor}{rgb}{0.0, 0.45, 0.85}   

\usepackage[
  colorlinks=true,
  linkcolor=mycitecolor,
  citecolor=mycitecolor,
  urlcolor=mycitecolor,
  hyperindex=true,
  linktocpage=true
]{hyperref}

\usepackage[capitalise,compress]{cleveref}
\usepackage{tcolorbox}

\newtcolorbox{shadedtheorem}{
  colback=gray!15,    
  colframe=white,     
  boxrule=0pt,        
  arc=0pt,            
  left=5pt,           
  right=5pt,
  top=5pt,
  bottom=5pt
}

\newtheoremstyle{upright-hang}
  {6pt}   
  {6pt}   
  {\normalfont\hangindent=1.5em\hangafter=1\relax} 
  {0pt}   
  {\bfseries} 
  {.}     
  {0.5em} 
  {}      

\theoremstyle{upright-hang}
\newtheorem{thm}{Theorem}
\numberwithin{thm}{section}
\newtheorem{cor}[thm]{Corollary}

\newtheorem{lem}[thm]{Lemma}

\newtheorem{prop}[thm]{Proposition}

\newtheorem{defn}[thm]{Definition}

\renewcommand{\thesection}{\arabic{section}}
\renewcommand{\thesubsection}{\thesection.\arabic{subsection}}
\renewcommand{\thesubsubsection}{\thesubsection.\arabic{subsubsection}}

\makeatletter
\renewcommand{\p@subsection}{}
\renewcommand{\p@subsubsection}{}
\makeatother

\makeatletter
\fancypagestyle{plain}{%
  \fancyhf{}%
  \fancyfoot[C]{\thepage}%
}
\fancypagestyle{revtexfooter}{%
  \fancyhf{}%
  \fancyfoot[C]{\thepage}%
}
\makeatother
\AtBeginDocument{\setstretch{1.25}\pagestyle{revtexfooter}}
\makeatletter
\renewcommand\frontmatter@title@format{\LARGE\bfseries\centering\parskip\z@skip}

\def\frontmatter@affiliationfont{\normalfont\selectfont}
\newcommand{\AndySectionStar}[1]{%
  \par\addpenalty\@secpenalty
  \addvspace{3.5ex \@plus 1ex \@minus .2ex}%
  \noindent{\normalfont\Large\bfseries #1\par}%
  \nobreak\vspace{2.3ex \@plus .2ex}%
  \@afterheading}
\newcommand{\AndySubsectionStar}[1]{%
  \par\addpenalty\@secpenalty
  \addvspace{3.25ex \@plus 1ex \@minus .2ex}%
  \noindent{\normalfont\large\bfseries #1\par}%
  \nobreak\vspace{1.5ex \@plus .2ex}%
  \@afterheading}
\newcommand{\AndySubsubsectionStar}[1]{%
  \par\addpenalty\@secpenalty
  \addvspace{3.25ex \@plus 1ex \@minus .2ex}%
  \noindent{\normalfont\normalsize\bfseries #1\par}%
  \nobreak\vspace{1.5ex \@plus .2ex}%
  \@afterheading}
\renewcommand\@seccntformat[1]{\csname the#1\endcsname\quad}
\def\@hangfrom@section#1#2#3{\@hangfrom{#1#2}#3}
\def\@hangfroms@section#1#2{#1#2}
\renewcommand\section{\@ifstar{\AndySectionStar}{\@startsection {section}{1}{\z@}%
  {-3.5ex \@plus -1ex \@minus -.2ex}%
  {2.3ex \@plus .2ex}%
  {\normalfont\Large\bfseries}}}
\renewcommand\subsection{\@ifstar{\AndySubsectionStar}{\@startsection{subsection}{2}{-\parindent}%
  {-3.25ex\@plus -1ex \@minus -.2ex}%
  {1.5ex \@plus .2ex}%
  {\normalfont\large\bfseries}}}
\renewcommand\subsubsection{\@ifstar{\AndySubsubsectionStar}{\@startsection{subsubsection}{3}{-\parindent}%
  {-3.25ex\@plus -1ex \@minus -.2ex}%
  {1.5ex \@plus .2ex}%
  {\normalfont\normalsize\bfseries}}}
\renewcommand*\l@section[2]{%
  \ifnum \c@tocdepth >\z@
    \addpenalty\@secpenalty
    \addvspace{0.75em \@plus\p@}%
    \begingroup
      \parindent\z@
      \rightskip\@pnumwidth
      \parfillskip-\@pnumwidth
      \leavevmode\bfseries
      \setlength\@tempdima{2.3em}%
      \advance\leftskip\@tempdima
      \hskip-\leftskip
      #1\nobreak\hfil\nobreak\hb@xt@\@pnumwidth{\hss #2}\par
    \endgroup
  \fi}
\renewcommand*\l@subsection{\@dottedtocline{2}{1.5em}{3.2em}}
\renewcommand*\l@subsubsection{\@dottedtocline{3}{3.8em}{4.1em}}
\makeatother

\usepackage{mathtools}
\tikzstyle{densely dashed}= [dash pattern=on 4pt off 3pt]

\usepackage{dsfont}

\allowdisplaybreaks

\begin{document}

\title{Spontaneous symmetry breaking at any temperature in a local one-dimensional model}

\author{Andrew Lucas}
\email{andrew.j.lucas@colorado.edu}
\affiliation{Department of Physics and Center for Theory of Quantum Matter, University of Colorado, Boulder, CO 80309, USA}

\begin{abstract} We present a one-dimensional classical statistical physics model on an infinite lattice with nearest-neighbor interactions and local state space $\mathbb{Z}^2$, with a ``helical" symmetry group $\mathbb{Z}^2\rtimes \mathbb{Z}$ mixing flavor symmetry with spatial translation.   For any $0<\beta<\infty$, we construct an uncountable infinity of DLR (Gibbs) states parameterized by $\mathbb{R}^2$, with a free $\mathbb{Z}^2$ flavor symmetry group action; as such, we interpret the model as exhibiting spontaneous symmetry breaking at any temperature.  The most general DLR state we construct is interpreted as an infinite sequence of localized domain walls separating a pair of asymptotic ``symmetry sectors".   Although any normalizable Gibbs state necessarily breaks a countably infinite freely-acting symmetry group, a natural regularization of the problem to (arbitrarily) large but finite systems with periodic boundary conditions has measure concentration in disjoint and far-separated symmetry-related clusters in the (now unique) Gibbs state, which justifies our interpretation of spontaneous symmetry breaking on the infinite line.  A modification of the model with local state space $\mathbb{R}^2$ must break a continuous helical symmetry at any temperature. This is not a contradiction with the Mermin-Wagner Theorem, because this continuous helical symmetry is non-compact and does not commute with translation.
\end{abstract}

\maketitle

\tableofcontents

\section{Introduction}
There has been recent interest in understanding the possible phase diagrams of statistical mechanics models with (countably) infinite on-site state spaces.  For example, \cite{Han:2025eiw} demonstrated simple classical lattice models in two spatial dimensions which spontaneously break a discrete symmetry at arbitrarily high temperature; see \cite{Huang:2025gvi,Tsao:2026hle,Hsin:2026ckp,Andriolo:2026udg,mehta2026} for other recent work on similar models.   In turn, this work was inspired by a long history of searching for (conformal) quantum field theories which spontaneously break a symmetry at every positive temperature \cite{Weinberg:1974highT,Chai:2020thermal,Chai:2020alltemperatures,Chai:2022persistentdiscrete,Chai:2022persistentcontinuous,Komargodski:2024zmt,Hawashin:2025ikp}.

This paper is about using countably infinite state spaces to break another common piece of lore in statistical mechanics: the absence of discrete spontaneous symmetry breaking in one spatial dimension \cite{vanhove}.  There is a rich history of trying to break this lore, and find thermal phase transitions in one dimension \cite{Poland:1966helixcoil,Kittel:1969zipper,Cuesta:2004onedimensional}.  Usually in these constructions, the state space has hard constraints and is not a direct product.  In this paper, we are interested in constructing thermal Gibbs measures for models on the infinite chain without such hard constraints.  Nevertheless, we wish to discover spontaneous symmetry breaking.  Mathematically, this can be done using the DLR formalism \cite{Dobrushin:1968randomfield,LanfordRuelle:1969observables,barrysimon,ny2007introductiongeneralizedgibbsmeasures,Georgii:2011gibbs}, which may be most famous for its use in establishing low-temperature ferromagnetism in the 2d Ising model on the infinite plane.  To achieve this goal, our model must have a (countably) infinite on-site state space, to avoid the standard rigorous no-go theorem forbidding thermal spontaneous symmetry breaking in any 1d model with finite/compact on-site state space and bounded local Hamiltonian \cite{barrysimon, Georgii:2011gibbs}.\footnote{It is possible to have thermal long-range order in a 1d Ising model with power-law interactions \cite{thouless,dyson}.  Therefore, the assumption of strict locality is important.}

We will succeed and build a model with a countably infinite on-site state space, with nearest-neighbor interactions in a one-dimensional chain, and spontaneous symmetry breaking at every  finite and strictly positive temperature. The model we introduce, which we call the $\mathrm{SL}(2,\mathbb{Z})$ \textbf{helical code}, has an on-site state space $\mathbb{Z}^2$ along with translation-invariant nearest-neighbor interactions.  It is invariant under a ``helical" $\mathbb{Z}^2\rtimes \mathbb{Z}$ symmetry which mixes an on-site flavor $\mathbb{Z}^2$ translation symmetry with spatial translation $\mathbb{Z}$.   We will construct an uncountably infinite number of thermal probability measures parameterized by $\mathbb{R}^2$, interpreted as a set of ``imperfect" domain wall states between a left and right asymptotic charge sector.  A somewhat related phenomenon in a very different continuous state space model was found in \cite{benfatto}.

This model was directly inspired by recent work on classical and quantum low-density parity-check codes in the Lee metric \cite{leecode1}: it is a very natural modification of the ``helical repetition code" from this paper which makes sense on an infinite lattice.   A similar model is also found in \cite{HuWatanabe2023clock}.

Our model is also very similar in spirit to a recent literature on ``modulated" symmetries in quantum condensed matter \cite{sala-PRL,Pace:2024tgk}, in particular the ``quantum breakdown" model described in \cite{lian1,lian2,lian3}.  In contrast to these papers from condensed matter physics, however, the model we describe here is classical, and we study finite-temperature symmetry breaking rather than quantum ground states.  But we expect that the mathematical methods we develop here may be interesting in constructing interesting (but pathological) models of quantum phases too; we conjecture some of what is possible in quantum helical codes in the outlook.

Although the idea that a local 1d model can spontaneously break a symmetry may seem quite surprising, we stress that it is not a new idea.  For example, non-compact massive bosons have DLR states associated with simple fluctuations around exponentially diverging background solutions \cite{benfatto,Dobrushin1980Gaussian,Friedli:2017lattice}.  More directly related,  \cite{Niedermaier:2003cq} focuses on SSB of the non-amenable non-compact group $\mathrm{SO}(1,2)$, and is most similar to our work.   Still, there are some important advantages of the helical code model, relative to \cite{Niedermaier:2003cq}.  We are able to construct conventional DLR (Gibbs) states, and give a more transparent physical interpretation for the distinct phases. As we are breaking a discrete symmetry (except for some brief remarks in Section \ref{sec:continuous}), we also uncover a potentially surprising phenomenon where -- on an infinite chain -- uncountably many Gibbs states, which must not all be related by a discrete symmetry, are found.  Lastly, the non-amenability of $\mathrm{SO}(1,2)$ is seen to be unimportant to the possibility of SSB; instead, the exponential divergence between on-site configurations in different symmetry sectors will be critical.  So, in light of these remarks, we think it is worth revisiting the question of 1d thermal SSB.

\section{Informal summary of results on the discrete model}\label{sec:informal}

This section is written for an physics audience, without mathematical precision, and summarizes the model and the intuition behind our main results, while the rest of the paper is more formal, until Section \ref{sec:continuous}.

\subsection{Overview of the model}

We study a statistical mechanics problem where degrees of freedom live on an infinite line (sites parameterized by $\mathbb{Z}$).   The degree of freedom on each site is a two-component integer-valued vector $x_i \in \mathbb{Z}^2$.   Pick any positive-definite symmetric hyperbolic matrix $A\in \mathrm{SL}(2,\mathbb{Z})$ -- namely, $A$ is an invertible $2\times 2$ matrix with integer coefficients and determinant 1 and one real eigenvalue $\lambda>1$ (the other is $\lambda^{-1}<1$); $A^{-1}$ is also symmetric and hyperbolic.   Then define the Hamiltonian (informally) \begin{equation}
    H = \sum_{n\in \mathbb{Z}} \left\lVert x_n - A x_{n+1} \right\rVert_1, \label{eq:H-informal}
\end{equation}
where the 1-norm $\lVert (x,y)\rVert_1 = |x|+|y|$.   

It is easy to see that this Hamiltonian is invariant under two symmetries:  \textbf{translation} \begin{equation}
    T\cdot (\ldots, x_n ,\ldots ) \rightarrow (\ldots, x_{n+1},\ldots),
\end{equation}
and a \textbf{helical flavor symmetry}: for any $a\in \mathbb{Z}^2$, \begin{equation}
   U_a\cdot  (\ldots, x_n ,\ldots) \rightarrow (\ldots, x_n + A^{-n}a, \ldots )
\end{equation}
These two group actions do not commute; $T$ and $U$ generate the helical symmetry $\mathbb{Z}^2 \rtimes \mathbb{Z}$.

Our goal is to understand the Gibbs state of this model: informally the probability distribution \begin{equation}
    \mu_\beta(x) \propto \mathrm{e}^{-\beta H(x)}. \label{eq:gibbs-informal}
\end{equation} On the infinite line this object does not exist, because \eqref{eq:H-informal} does not really exist.  The formalism of DLR states \cite{Dobrushin:1968randomfield,LanfordRuelle:1969observables,barrysimon,ny2007introductiongeneralizedgibbsmeasures,Georgii:2011gibbs} will be used to define all of this more precisely in later sections.  

In ordinary statistical mechanics, one expects that at \emph{any} $0<\beta<\infty$ there is a \textbf{unique} DLR state in a local one-dimensional model.  This is the usual intuition that there are no non-trivial phases at non-zero temperature in one dimension.  This statement can be proven rigorously for statistical mechanics models with finite-dimensional on-site state spaces with bounded and extensive Hamiltonians.  In contrast, in two or more spatial dimensions, there are often multiple DLR states at low temperature.  For example in the Ising model, there is a separate DLR state for the mostly spin-up vs. mostly spin-down symmetry sectors of the spontaneously broken $\mathbb{Z}_2$ symmetry.  A physical interpretation of this is that we can construct the DLR states by pinning boundary conditions on a large $L\times L$ square to be either all up or all down.  Even near the center of the square, the local spins feel the influence of the boundary conditions in the long-range ordered phase.     One proves  that there are two limiting probability measures $\mu^\uparrow_\beta $ and $\mu^\downarrow_\beta $, each of which \textbf{locally} look like \eqref{eq:gibbs-informal}, but with a global ``memory" of the boundary conditions.

Let us give two (somewhat related) arguments why the countably infinite state space in our model might break this intuition.  Firstly, as we make precise in Theorem \ref{thm:sound}, if we do pin boundary conditions in a finite-size domain, the model is extremely ``rigid".  For example, if we pin boundaries $x_{\pm L} = 0$, then the energy cost of a configuration $(x_{-L+1},\ldots, x_{L+1})$ obeys \begin{equation}
    H \ge \eta \sum_{j=-L+1}^{L-1} \lVert x_j\rVert_1 \label{eq:intuitive-sound}
\end{equation}for some fixed $\eta>0$.  This is suggestive that the model is \emph{pinned} to a specific symmetry sector where $x_n\approx 0$.  This ought to be surprising at first glance!  There is also a symmetry sector where $x_n\approx A^{-n}a$ for any $a\in \mathbb{Z}^2$ -- why do we not  see fluctuating domains between these different sectors?  In a usual 1d system, the reason these fluctuating domains dominate (and lead to a single DLR state) is that the energy cost of a domain wall is finite: $E_{\text{dw}} \sim 1$, while the entropy of where to place the domain wall is $S_{\text{dw}}\sim \log L$.  The free energy cost is therefore \begin{equation}
    F_{\text{dw}} \sim E_{\text{dw}} - \beta^{-1}S_{\text{dw}} \lesssim -\log L,
\end{equation}
and it always pays to introduce domain walls.  In contrast, \eqref{eq:intuitive-sound} suggests that inserting a domain wall of one symmetry sector inside of another costs $E_{\text{dw}} \sim \exp(L)$, due to the exponentially diverging trajectories ($A$ is hyperbolic):  this energy cost can \textbf{vastly} outweigh the entropic gain.  
Clearly this is only possible in a local model if we rely on the infinite state space.   What is non-trivial is that the Hamiltonian looks ``mild": changing any single $x_j$ to a nearby site in $\mathbb{Z}^2$ has a bounded energy cost.  The fact that the helical symmetry penalizes nucleating a domain wall in an ambient phase very aggressively is what permits spontaneous symmetry breaking.

There is a second motivation for the existence of multiple Gibbs states, which is that it was established long ago in the mathematics literature that Markov random fields (a.k.a. local one-dimensional Hamiltonians) can have non-unique Gibbs states in 1d \cite{Spitzer:1975phase,Kesten:1976countable,Cox:1977phase,Zachary:1983countable}.  Again, these models rely on countably infinite local state spaces.  An intuitive justification for these results (which also helsp to explain our findings) is as follows:  \textbf{after conditioning} a Markov chain on both an initial and final condition, the path that the Markov chain followed can become very rigid.  In the context of our model, we will see extremely explicitly how this intuition is realized, in an example where this rigidity can be understood in the physicist's language of spontaneous symmetry breaking.  We will justify this in Section \ref{sec:finite-size} by showing that in finite-sized systems, the unique Gibbs state concentrates onto distinct clusters in state space related by symmetry, similar to the case of an Ising ferromagnet in a large (but finite) $L\times L$ square lattice, described above.

\subsection{Structure of the DLR states}\label{sec:informal-structure}

Combining these two intuitions, let us now sketch our main result:  there is a DLR state $\mu^c_\beta$ for every $c\in \mathbb{R}^2$.  $\mu^c_\beta$ is equivariant (transforms naturally) under the helical symmetry.  In a nutshell, each $\mu^c_\beta$ describes fluctuations around a local solution of the form \begin{equation}
 x_n \approx  A^{-n} c.
\end{equation}
To better understand the implications of this, suppose that $Au=\lambda u$ and $Av = \lambda^{-1}v$ are the two orthogonal eigenvectors of $A$;  writing $c = c_u u + c_vv$ we see that \begin{equation}
    x_n \approx c_u \lambda^{-n} u + c_v \lambda^n v. \label{eq:informal-xn-approx}
\end{equation}
As $n\rightarrow -\infty$, the solution diverges exponentially (if $c_u\ne 0$) and approaches one solution labeled by $c_u$; as $n \rightarrow +\infty $, the solution again diverges exponentially if $c_v \ne 0$.  In a nutshell, the typical solution describes a kind of domain wall between the $c_u$ charge sector on the left and the $c_v$ charge sector on the right.  Only very special pairs of $(c_u,c_v)$ are related to the simplest $(0,0)$  choice by the discrete helical symmetry group action; typical choices of $(c_u,c_v)$ are therefore \emph{unrelated} to each other by any symmetry.  This is somewhat reminiscent of the literature on spin glasses, where there are many low-energy sectors unrelated by any broken symmetry \cite{Mezard:1987spinglass}.  

However, one should \emph{not} interpret this phenomenon as ``glassy".  Instead, we can understand it as follows.   For simplicity in the remainder of the paragraph, we take $c_u=0$ and $c_v \ne 0$.  In this case, we have $x_n \approx0$ as $n\rightarrow -\infty$; in this case, the system looks like it is in an exact ``symmetry sector" -- an exact ground state of $H$ has $x_n=0$ for all $n$.   However, as we approach sufficiently large $n$ (depending on $c_v$ and $\lambda$), we see the solution begin to deviate and grow exponentially with increasing $n$.   For the sake of argument, suppose that \begin{equation}
    c_v = \frac{a_1}{\lambda^{m_1}} + \frac{a_2}{\lambda^{m_2}} + \frac{a_3}{\lambda^{m_3}} + \cdots
\end{equation}
where each $\lambda^{-1} \lesssim |a_i| \lesssim \lambda$, and each $m_i<m_{i+1}$ is an integer.  For illustrative purposes, suppose that there are large gaps between each $m_i$.   Then notice that for $n \ll m_1$, we will have $x_n\approx 0$, but once $m_1\lesssim n_1 \ll  m_2$, we will have a ``slip" to $x_n \approx a_1 \lambda^{n-m_1}\cdot v$.   The next such ``slip" will occur when $n\approx m_2$, and so on.  The infinite possible locations and magnitudes of such slips are all encoded in $c_v \in \mathbb{R}$.   

Of course, we can just as easily do the same analysis moving leftwards to $n\rightarrow -\infty$ -- that data is stored in $c_u$.  The combination of $(c_u,c_v) = c\in \mathbb{R}^2$ label all of the DLR states we have constructed; their existence is formally established by Theorem \ref{thm:exist}.

\eqref{eq:informal-xn-approx} is quantified by Corollary \ref{cor:deviation}.  As suggested above, the proof of this corollary ultimately stems from our estimate \eqref{eq:intuitive-sound}.  Even when $c$ is a generic real vector, the system simply ``knows" where to place all the slips, as in the discussion of the previous paragraph, even though the model only has nearest-neighbor interactions in 1d.  If the slips are in the wrong place, we incur an enormous energy penalty from \eqref{eq:intuitive-sound}.

It seems plausible to conjecture that these are all the extremal DLR states in this model.  But given the surprises with infinity in statistical mechanics, it is also plausible if this conjecture is incorrect.  The conjecture would be incorrect if one may choose sequences of pathological boundary conditions which do not diverge according to eigenvectors of $A$, but for which the resulting Gibbs states in the ``middle" of these increasingly large chain models nevertheless converge (to a distribution that is not a convex sum of the DLR states we already furnished).  We leave further analysis to future work.

\section{Formal summary of results on the discrete model}
\subsection{The model}
We now turn to a more formal analysis of these models with helical symmetry, beginning with setting up some definitions and conventions for the rest of the paper.  
\begin{defn}[Helical $\mathrm{SL}(2,\mathbb{Z})$ symmetry]
    Consider the state space $\Omega = (\mathbb{Z}^2)^{\mathbb{Z}}$, whose elements are labeled by vectors with components $x_n \in \mathbb{Z}^2$ ($n\in \mathbb{Z}$).    Pick a symmetric hyperbolic\footnote{$A$ has two real eigenvalues, one inside and one outside the unit circle in $\mathbb{C}$.} positive matrix $A\in \mathrm{SL}(2,\mathbb{Z})$.  A helical $\mathbb{Z}^2\rtimes \mathbb{Z}$ group action on $\Omega$ is generated by \begin{subequations}\label{eq:helical-group-action}\begin{align}
       U_a \cdot  x_n &:= x_n + A^{-n}a, \\
       T \cdot x_n &:= x_{n+1} 
    \end{align}\end{subequations}
    for all $n$, and any $a\in \mathbb{Z}^2$.  Note that $U_a U_b = U_{a+b}$, while 
    \begin{equation}
    T \cdot U_a = U_{A^{-1} a} \cdot T, \label{eq:helicalcommutation}
\end{equation}  A function has a helical symmetry if it is invariant under \eqref{eq:helical-group-action}.
\end{defn}

\noindent In the discussion that follows, it will be useful to introduce some notation related to the choice of hyperbolic matrix $A$, which we collect below: 
\begin{defn}[Eigenvectors of $A$] \label{def:eigenvectors}
Hyperbolic positive-definite symmetric $A \in \mathrm{SL}(2,\mathbb{Z})$ has real eigenvalues $\lambda$ and $\lambda^{-1}$ for some $\lambda > 1$.  We choose orthonormal eigenvectors such that: \begin{subequations} \label{eq:A-eigenvectors}
    \begin{align}
        A u := \lambda u, \\
        Av := \lambda^{-1}v,
    \end{align}
\end{subequations}and $ u^{\mathsf{T}} u  =v^{\mathsf{T}}v=1$.
    
\end{defn}

\begin{prop}
    $\lambda$ and $\lambda^{-1}$ are irrational; moreover, $u,v \notin \mathbb{Q}^2$, and $u\cdot v=0$.
\end{prop}
\begin{proof}
    Suppose that \begin{equation}
        A = \left(\begin{array}{cc} a &\ b \\ b &\ d \end{array}\right)
    \end{equation}
    with $\det(A) = ad-b^2=1$.  Since $A$ is symmetric, its eigenvectors are orthogonal: $u\cdot v = 0$.   By assumption $A$ has real eigenvalues, and since the determinant is the product and the sum is the trace, we know $a+d = \lambda + \lambda^{-1}$ which implies that \begin{equation}
        \lambda = \frac{a+d \pm \sqrt{(a+d)^2-4}}{2}.
    \end{equation}
    The fact that $A$ is hyperbolic means that $|a+d|>2$:  if $a+d\in \lbrace -1,0,1\rbrace$ then $\lambda$ is complex, while if $a+d\in \lbrace \pm 2\rbrace$ then both eigenvalues are real but  lie on the unit circle in $\mathbb{C}$.   If $|a+d|$ is an integer larger than 2, then $(a+d)^2 - 4$ is an integer that is not a perfect square, meaning that $\sqrt{(a+d)^2-4}$ is irrational.  \eqref{eq:A-eigenvectors} is then inconsistent unless at least one component of $u$ and $v$ is an irrational number (as both $A$ and $u$ would be $\mathbb{Q}$-valued, $\lambda$ would itself need to be rational).  
\end{proof}

We are now ready to formally introduce the helical linear code as a family of unbounded Hamiltonian functions defined on arbitrary subsets of the 1d chain $\mathbb{Z}$:

\begin{defn}[Helical linear code]\label{def:helicalcode}
    Let $\Lambda = \lbrace j_-, j_-+1,\ldots, j_+\rbrace  \subset \mathbb{Z}$ be a finite subset.  Define the local functions \begin{equation}
        h(x_j, x_{j+1}) := \lVert x_j - A x_{j+1}\rVert_1,
    \end{equation}which are manifestly invariant under \eqref{eq:helical-group-action}.  For any $\Lambda$, define $H_{\Lambda,b_-,b_+}:\Omega \rightarrow \mathbb{R}$ by \begin{equation} \label{eq:H-Lambda}
        H_{\Lambda,b_-,b_+}(x) := \sum_{j=j_-}^{j_+-1}h(x_j,x_{j+1})
    \end{equation}
    with $x_{j_-} = b_-$ and $x_{j_+}=b_+$ held fixed.  
    The helical symmetry acts equivariantly on $H_{\Lambda,b_-, b_+}$: \begin{equation}
     H_{\Lambda, b_-, b_+}(x) = H_{\Lambda, b_- + A^{-j_-} a, b_+ + A^{-j_+}a}(U_a\cdot x).
    \end{equation} 
\end{defn}
\noindent For the manipulations that follow, it is useful to define an inner product space $\mathrm{L}^2(\mathbb{Z}^2)$ equipped with the standard inner product \begin{equation}
    \langle a|a^\prime\rangle = \delta_{a,a^\prime}.
\end{equation}We will use the bra-ket notation for vectors on this auxiliary space, and we define the non-normalizable vector\footnote{$|1\rangle$ is not an element of $\mathrm{L}^2(\mathbb{Z}^2)$, but we can still use it to streamline some notation below as long as we only evaluate inner products with vectors whose 1-norm is finite.} \begin{equation}
    |1\rangle := \sum_{a \in \mathbb{Z}^2} |a\rangle .
\end{equation}    We further define a transfer matrix on $\mathcal{V}$: \begin{defn}[Transfer matrix]\label{def:transfer-matrix}
        We define the \textbf{transfer matrix} $K_\beta$ entry-wise: \begin{equation}
         \langle x_{j+1}| K_\beta|x_j\rangle := \mathrm{e}^{-\beta h(x_j,x_{j+1})}.
    \end{equation}
\end{defn}

\subsection{Gibbs states in a finite region}
Having introduced the model above, the simplest thing to do is to study its behavior in a finite-sized chain: e.g. restrict $\Omega$ to $(\mathbb{Z}^2)^\Lambda$ for some finite $\Lambda \subset \mathbb{Z}$.  Indeed, as is standard \cite{ny2007introductiongeneralizedgibbsmeasures}, careful analysis of such finite-size measures will be necessary to take the genuine thermodynamic limit eventually.  Therefore, we begin by collecting some simple facts about the probability measures on finite-sized systems.  The first basic observation is that all such measures must break the helical symmetry.
\begin{prop}
    If $\mu$ is a probability measure on $\Omega$, then $\mu$ is not invariant under $U_a$ for any $a\ne 0$.
\end{prop}
\begin{proof}
    Consider $S_b= \lbrace x \in \Omega : x_0=b\rbrace$.  Notice that $U_a\cdot S_b = S_{b+a}$. Since $\mu$ is a probability measure, we know that \begin{equation}
        \mu(\Omega) = 1  =\sum_{a\in \mathbb{Z}^2}\mu (S_a) \label{eq:muOmega}
    \end{equation}
    If $(U_a)_*\mu = \mu$, then all $\mu(S_{na+b})=\mu(S_b)$ are the same.  If (for any fixed $b$) $\mu(S_b)>0$, the sum \eqref{eq:muOmega} diverges due to the subset $a\mathbb{Z}+b\subset\mathbb{Z}^2$, so $\mu(S_b)=0$ is required for every $b$.  But then the sum in \eqref{eq:muOmega} would also apply $\mu(\Omega)=0$, contradicting $\mu(\Omega)=1$.
\end{proof}

The non-trivial question becomes whether it is possible to have a probability measure $\mu$ which ``looks like" a thermal state, on which the helical symmetry acts equivariantly.  The existence of such a measure is far from obvious -- for example, there is no meaningful ``uniform" probability measure over the integers, which can be understood physically as the statement that there is no steady-state distribution of a (lazy) random walk on $\mathbb{Z}$.  We will, in contrast, show the existence of well-defined measures $\mu$ for the helical linear code, on which $U_a$ acts equivariantly.  

We begin by considering the Gibbs measure on finite-size and bounded domains $\Lambda \subset \mathbb{Z}$:

\begin{prop}[Gibbs state on $\Lambda$]\label{prop:finite-gibbs}
The \textbf{Gibbs state} \begin{equation}\label{eq:finite-gibbs}
    \mu_{\Lambda,\beta,b_-, b_+}(x):= \frac{\mathrm{e}^{-\beta H_{\Lambda,b_-,b_+}(x)}}{Z_{\Lambda,\beta,b_-,b_+}} = \frac{\langle b_+|K_\beta|x_{j_+-1}\rangle \langle x_{j_+-1}|K_\beta |x_{j_+-2}\rangle \cdots \langle x_{j_-+1}|K_\beta |b_-\rangle}{\langle b_+| K_\beta^{j_+-j_-}|b_-\rangle } .
\end{equation}
is a well-defined probability measure on $(\mathbb{Z}^2)^\Lambda$,\footnote{Of course, the measure pins the degrees of freedom at the endpoints of $\Lambda$.} i.e. $Z_{\Lambda,\beta,b_-,b_+} := \langle b_+| K_\beta^{j_+-j_-}|b_-\rangle <\infty$.   Similarly, the following one-sided Gibbs measures also exist: \begin{subequations}
    \begin{align}
        \mu^+_{\Lambda,\beta,b_-}(x) &:= \frac{ \langle x_{j_+}|K_\beta |x_{j_+-1}\rangle \cdots \langle x_{j_-+1}|K_\beta |b_-\rangle}{\langle 1| K_\beta^{j_+-j_-}|b_-\rangle }, \\
        \mu^-_{\Lambda,\beta,b_+}(x) &:=\frac{\langle b_+|K_\beta|x_{j_+-1}\rangle \cdots \langle x_{j_-+1}|K_\beta |x_{j_-}\rangle}{\langle b_+| K_\beta^{j_+-j_-}|1\rangle }.
    \end{align}
\end{subequations}
\end{prop}
\begin{proof}
    Provided in Section \ref{sec:finite-gibbs-exists}.
\end{proof}
The Gibbs state is a 1d Markov random field, meaning that the measure restricted to any connected subregion $\Lambda_0 = \lbrace k_-,\ldots, k_+\rbrace$ depends only on the state at the boundary:  
\begin{prop} \label{prop:mrf}
    $\mu_{\Lambda,\beta,b_-, b_+}$ is a Markov random field: for any $\Lambda_0 \subset \Lambda$, \begin{align} \label{eq:mu_beta}
        \mu_{\Lambda,\beta,b_-, b_+}&(x_{j\in \Lambda_0} | x_{k_-},x_{k_+}) \notag \\
        &= \dfrac{\displaystyle \langle b_+|K_\beta^{j_+ - k_+}|x_{k_+}\rangle \cdot \prod_{i=k_-}^{k_+-1} \langle x_{i+1}|K_\beta |x_i\rangle \cdot \langle x_{k_{-}}|K_\beta^{k_--j_-}|b_-\rangle}{\langle b_+|K_\beta^{j_+ - k_+}|x_{k_+}\rangle \langle x_{k_+}|K_\beta^{k_+-k_-}|x_{k_-}\rangle\langle x_{k_{-}}|K_\beta^{k_--j_-}|b_-\rangle} \notag \\
        &= \mu_{\Lambda_0,\beta,x_{k_-},x_{k_+}}(x_{j\in \Lambda_0}).
    \end{align}
    An analogous conclusion applies when $\Lambda_0$ is disconnected, albeit data at each boundary is required.
\end{prop}

\noindent  This property is obvious from the definition of the Gibbs state, but it is very important nonetheless.  We now consider a limiting process where we attempt to take the subregion $\Lambda$ to be larger and larger, while restricting observables to a fixed sized region $\Lambda_0 \subset \Lambda$.  The more important question is: if $x_{k_-}$ and $x_{k_+}$ are not frozen, does an unconditional analogue of \eqref{eq:mu_beta} still exist when we look at finite subsets $\Lambda_0$, and take $\Lambda$ to be large?  The probability that we see a given set of boundary conditions is \begin{align} 
    \mu_{\Lambda,\beta,b_-, b_+}&(x_{k_-},x_{k_+}) = \frac{\langle b_+|K_\beta^{j_+-k_+}|x_{k_+}\rangle\langle x_{k_+}|K_\beta^{k_+-k_-}|x_{k_-}\rangle \langle x_{k_-}|K_\beta^{k_--j_-}|b_-\rangle}{\langle b_+|K_\beta^{j_+-j_-}|b_-\rangle}. \label{eq:mu_endpoints}
\end{align}
Fixing $k_\pm$, the goal is now to construct a sequence of growing $\Lambda$, and $b_\pm$ (which will in general diverge) such that  $\mu_{\Lambda,\beta,b_-, b_+}(x_{k_-},x_{k_+})$ nevertheless converges to a well-posed limit as $|\Lambda|\rightarrow \infty$. If that limit exists, then clearly by Proposition \ref{prop:mrf}, the sequence of $\mu_{\Lambda,\beta, b_-,b_+}(x_{j\in \Lambda_0})$ would converge as well.  

\subsection{$\sigma$-algebras and DLR measures}
In order to take the infinite volume limit described just now, it may be helpful to review a more formal probabilist's perspective on the problem.  The discussion that follows is essentially unnecessary to understand the rest of the paper, except for Theorem \ref{thm:extremal}, which establishes that in a certain mathematical sense, the many DLR (Gibbs) states that we will construct in Theorem \ref{thm:exist} should be understood as distinct phases of matter.  In the discussion that follows it is useful to introduce the set of all finite subsets of $\mathbb{Z}$ as \begin{equation}
        \mathcal{S} := \lbrace \Lambda \subset \mathbb{Z} : |\Lambda|<\infty\rbrace.
    \end{equation}  

Formally, our goal is to construct a probability measure $\mu$ on $(\Omega,\mathcal{F})$.  To define the desired $\sigma$-algebra $\mathcal{F}$, it is useful to first construct the finite-size $\sigma$-algebra \begin{equation}
    \mathcal{F}_\Lambda := \sigma( \lbrace x_j : j\in \Lambda \rbrace).
\end{equation}
Conditioning on $\mathcal{F}_\Lambda$ means that  the variables $x_j$ are known for each $j\in \Lambda$: thus $\mu(A|\mathcal{F}_\Lambda): (\mathbb{Z}^2)^\Lambda\rightarrow [0,1]$ denotes a function on the known $\lbrace x_j:j\in \Lambda\rbrace$.   $f: \Omega \rightarrow \mathbb{R}$ is $\mathcal{F}_\Lambda$-measurable if and only if $f(\lbrace x_j : j\in \Lambda\rbrace)$ only depends on degrees of freedom in $\Lambda$.   With the definition of $\mathcal{F}_\Lambda$ in hand, we may now define  \begin{equation}
    \mathcal{F} = \sigma \left( \bigcup_{\Lambda \in \mathcal{S}} \mathcal{F}_\Lambda \right)
\end{equation}
as the smallest $\sigma$-algebra containing all $\sigma$-algebras $\mathcal{F}_\Lambda$.  With this notation in mind, we can introduce the formal DLR equation, which says in a nutshell that the infinite-volume measure $\mu_\beta$ should almost surely reduce to the local Gibbs state \eqref{eq:finite-gibbs} in any region $\Lambda$, given $\mathcal{F}_{\Lambda^{\mathrm{c}}}$: 

\begin{defn}[DLR measure]
    A \textbf{DLR measure} $\mu_\beta$ at inverse temperature $\beta$ is a probability measure on $(\Omega,\mathcal{F})$ such that for all $A\in \mathcal{F}$, $\mu_\beta$-almost surely, for any $\Lambda \in \mathcal{S}$, \begin{equation} \label{eq:DLRdef}
        \mu_\beta(A|\mathcal{F}_{\Lambda^{\mathrm{c}}})(\omega_{\Lambda^{\mathrm{c}}} ) = \mu_{\Lambda ,\beta}(A | \omega_{\Lambda^{\mathrm{c}}}).
    \end{equation}
    where $\omega_{\Lambda^{\mathrm{c}}} := \lbrace x_j : j\notin \Lambda\rbrace$, and we define, letting $H_\Lambda$ denote the finite-region Hamiltonian \eqref{eq:H-Lambda} with unrestricted boundary conditions,  \begin{equation}
        \mu_{\Lambda,\beta}(A|\omega_{\Lambda^{\mathrm{c}}})  \propto  \sum_{x\in (\mathbb{Z}^2)^\Lambda} \mathbb{I}[(x,\omega_{\Lambda^{\mathrm{c}}}) \in A] \mathrm{e}^{-\beta H_\Lambda(x,\omega_{\Lambda^{\mathrm{c}}})}.
    \end{equation}
\end{defn}

\noindent The set of DLR measures is convex, and one is typically interested in extrema of this set.  For example, in the ferromagnetic phase of the 2d Ising model, there are two such extremal DLR states $\mu^\uparrow$ and $\mu^\downarrow$ corresponding to the dominant spin orientation.\footnote{Any mixture $p\mu^\uparrow + (1-p)\mu^\downarrow$ is also a DLR state.  Therefore, we will often talk about classifying DLR states, but we really mean classifying the \emph{extrema} of the convex set of DLR states.}  Intuitively this means that they correspond to a well-behaved choice of boundary condition.  To make this more formal, we introduce:
\begin{defn}[Tail $\sigma$-algebra]
    The \textbf{tail $\sigma$-algebra} is \begin{equation}
        \mathcal{F}_\infty := \bigcap_{\Lambda \in \mathcal{S}} \mathcal{F}_{\Lambda^{\mathrm{c}}}.
    \end{equation}
    $\mathcal{F}_\infty$ is a precise notion of the ``boundary conditions at $\infty$".  
\end{defn}
\noindent We then state the following standard results \cite{ny2007introductiongeneralizedgibbsmeasures} without proof.  One might worry that the countably infinite on-site state space introduces further subtlety, but this is not true: it is sufficient that $\mathbb{Z}^2$, and hence $\Omega$, is a Polish space.
\begin{thm}[Extremal DLR state] \label{thm:extremal-property}
    DLR state $\mu_\beta$ is extremal if and only if for every $B\in \mathcal{F}_\infty$, $\mu_\beta(B) \in \lbrace 0 , 1\rbrace$.  Two extremal DLR states $\mu_\beta$ and $\mu^\prime_\beta$ are distinct if and only if there exists an event $B^\prime \in \mathcal{F}_\infty$ such that $\mu_\beta(B^\prime)=0$ and $\mu_\beta^\prime(B^\prime)=1$. 
\end{thm}
\begin{cor} \label{cor:extremal-property}
    If and only if $\mu_\beta$ is extremal, for every $\Lambda \in \mathcal{S}$ and for every bounded $\mathcal{F}_\Lambda$-measurable function $f$, $\mu_\beta$-almost surely, \begin{equation} \label{eq:extremal-property}
        \mu_\beta(f  | \mathcal{F}_\infty) = \mu_\beta(f).
    \end{equation} 
\end{cor}
\noindent The intuition behind Corollary \ref{cor:extremal-property} is as follows.  The left hand side of \eqref{eq:extremal-property} is a function of boundary data, and generically there exists some function $f$ (an ``order parameter") that can tell the difference between two distinct phases (i.e. extremal DLR states).  But the right hand side of \eqref{eq:extremal-property} is a constant.  The only way that this equation can hold $\mu_\beta$-almost surely is for $\mu_\beta(B)\in \lbrace 0,1\rbrace$ for any $B\in\mathcal{F}_\infty$, as asserted in Theorem \ref{thm:extremal-property}.

\subsection{Existence of infinitely many extremal DLR states}\label{sec:mucbeta}
We now return to the helical code model, and construct DLR states for this model.  A useful aid in this regard is: 
\begin{defn}[Rounding map]
    Define a \textbf{rounding map} $\mathcal{R} : \mathbb{R}^2 \rightarrow \mathbb{Z}^2$ such that \begin{equation}
        \sup_{y\in \mathbb{R}^2} \lVert \mathcal{R}(y) - y \rVert_1 \le \mathcal{R}_{\text{tol}} \label{eq:Rtol}
    \end{equation}
    for any O(1) number $\mathcal{R}_{\text{tol}} \ge 1$.
\end{defn}
\noindent A useful intuitive way to think about constructing the DLR states is to consider a special ``helical sequence" of boundary conditions, using the rounding map. 
\begin{defn}[Helical sequence]\label{def:helical-sequence}
    Pick $c\in \mathbb{R}^2$.  A \textbf{helical sequence} corresponds to a sequence of $(\Lambda_N,b_-^{(N)}, b_+^{(N)})$ defined as follows.  Pick any $N_0$ such that $\sup_{j \in \Lambda_0} |j| < N_0$.  Then, for each $N\in \lbrace N_0,N_0+1,\ldots\rbrace$, \begin{equation}
        \Lambda_N := \lbrace -N, -N+1,\ldots, N-1,N\rbrace
    \end{equation}(clearly $\Lambda_0\subset \Lambda_N$) and boundary conditions\begin{subequations}
        \begin{align}
            b_-^{(c,N)} &:= \mathcal{R}\left(A^{N} c \right), \\
            b_+^{(c,N)} &:= \mathcal{R}\left(A^{-N} c \right).
        \end{align}
    \end{subequations}
\end{defn}
\begin{thm}\label{thm:exist}
    For every $c\in\mathbb{R}^2$, there exists a distribution $\mu^{c}_\beta$ on $\Omega$ for which the limit \begin{equation} \label{eq:thm_exist_main}
        \lim_{N\rightarrow\infty } \left\lVert  \mu_{\Lambda_N,\beta, b_-^{(c,N)}, b_+^{(c,N)}}(x_{j\in \Lambda_0}) - \mu^{c}_\beta (x_{j\in\Lambda_0})\right\rVert_{\mathrm{TV}} = 0.
    \end{equation}
    exists.  In particular, $\mu^{c}_\beta$ is the limit of expectation values in a helical sequence and is independent of $\mathcal{R}$.  $\mu^c_\beta$ is equivariant under the helical symmetry group: \begin{subequations}\label{eq:equivariant}\begin{align}
    T_*  \mu^c_\beta&= \mu^{A^{-1}c}_\beta,\\ 
    (U_a)_* \mu^c_\beta &= \mu^{c+a}_\beta.
\end{align}
\end{subequations}
$\mu^c_\beta$ obeys the DLR equation \eqref{eq:DLRdef}:  for any $p\le q$, and $\mu^c_\beta$-almost every $(x_{p-1},x_{q+1})$, \begin{equation} \label{eq:DLR}
        \mu^c_\beta(x_p,\ldots, x_q | x_{p-1},x_{q+1}) = \mu_{\lbrace p-1,p,\ldots, q+1\rbrace,\beta,x_{p-1},x_{q+1}}(x_p,\ldots, x_q).
    \end{equation}
\end{thm}
\begin{proof}
    Contained in Section \ref{sec:thmproof}.
\end{proof}

\begin{cor}\label{cor:continuity}
    For any fixed finite region $\Lambda_0=\lbrace p,\ldots, q\rbrace$, the marginal distribution $\mu^c_\beta(x_{j\in \Lambda_0})$ is continuous in $c$ in total variation.
\end{cor}
\begin{proof}
     Provided in Section \ref{sec:corproof-2}.
\end{proof}

An intuitive understanding of the meaning of the DLR states $\mu^c_\beta$ is provided in Section \ref{sec:informal-structure}: at any fixed site $j$, we had argued that $x_j$ would stay reasonably close to the rounded ideal solution $\mathcal{R}(A^{-j}c)$.  Having established the DLR state $\mu^c_\beta$, we can now quantitatively verify that this is the correct interpretation:
\begin{cor} \label{cor:deviation}
    There exist constants $0<\alpha, J < \infty$ such that for every $j\in\mathbb{Z}$ and $x\in\mathbb{Z}^2$, \begin{equation} \label{eq:exponential-decay}
        \mu^c_\beta \left( x_j = x \right) < J \mathrm{e}^{-\alpha  \left\lVert x - \mathcal{R}(A^{-j}c)  \right\rVert_1 }.
    \end{equation}
\end{cor}
\begin{proof}
    Provided in Section \ref{sec:corproof}.
\end{proof}
\begin{cor}
    If $c=c_u u +c_vv$ with $u$ and $v$ the eigenvectors of $A$ (Definition \ref{def:eigenvectors}), then $\mu_\beta^c$-almost surely \begin{subequations}
        \begin{align}
            \lim_{n\rightarrow -\infty} \lambda^n u\cdot x_n &= c_u, \label{eq:cu-limit} \\
            \lim_{n\rightarrow +\infty} \lambda^{-n} v\cdot x_n &= c_v. \label{eq:cv-limit}
        \end{align}
    \end{subequations}
\end{cor}
\begin{proof}
    We prove \eqref{eq:cv-limit}, as \eqref{eq:cu-limit} follows from an identical argument.  Define the events \begin{equation}
        E_n := \left\lbrace \lambda^{-n} \left|v \cdot \left(x_n - \mathcal{R}\left(A^{-n} c\right) \right) \right| > \epsilon \right\rbrace .
    \end{equation}   Corollary \ref{cor:deviation} implies that, for sufficiently large $n$, \begin{equation}
        \mu_\beta^c \left( E_n  \right) < J \sum_{z : \lVert z \rVert_1 \ge \lfloor \epsilon \lambda ^n\rfloor/2} \mathrm{e}^{-\alpha \lVert z\rVert_1} < J \mathrm{e}^{-\epsilon \lambda ^n/4}.
    \end{equation}
    Since $\lambda >1$, the Borel-Cantelli Lemma shows that $\mu^c_\beta$-almost surely, only a finite number of events $E_n$ are realized for any $\epsilon>0$.  Hence the limit \eqref{eq:cv-limit} exists and equals $c_v$.
\end{proof}
\noindent This corollary clearly also implies that each $\mu^c_\beta$ is a distinct DLR state, since the at least one of the two limits \eqref{eq:cu-limit} and \eqref{eq:cv-limit} differ for any $c\ne c^\prime$.   This rigorously justifies the claim that we have constructed uncountably many DLR states.  The last theorem  establishes that these DLR states are not merely superpositions of a countable set of DLR states; instead, each of them represent genuinely distinct extremal DLR states (i.e. thermodynamic phases).

\begin{thm}\label{thm:extremal}
    For every $c\in \mathbb{R}^2$, the DLR state $\mu^c_\beta$ is extremal.
\end{thm}
\begin{proof}
    Proved in Section \ref{sec:extremal-proof}.
\end{proof}

\section{Existence of a Gibbs state in a finite region}
In order to prove Theorem \ref{thm:exist}, we need to develop a number of technical tools to handle the countably infinite on-site state space.  In this section, we establish a number of useful energetic bounds towards this end. In the sections that follow, we often omit the explicit $c$-dependence that arises in many definitions and formulae to avoid clutter, writing $c$ dependence only when we wish to emphasize it (or compare two different $c$).

\subsection{Proof of Proposition \ref{prop:finite-gibbs}}\label{sec:finite-gibbs-exists}
Firstly, we must establish the existence of the Gibbs state.  This is very straightforward, but we delayed the proof until now, as it ``thematically" foreshadows many further properties we'll develop.  

The key observation is that since $A\in \mathrm{SL}(2,\mathbb{Z})$ is a bijective map on $\mathbb{Z}^2$, we have the following identity for any $a\in \mathbb{Z}^2$: \begin{align}
    \langle 1|K_\beta|a\rangle &= \sum_{x\in\mathbb{Z}^2}\mathrm{e}^{-\beta \lVert a-Ax\rVert_1} = \sum_{x\in\mathbb{Z}^2}\mathrm{e}^{-\beta \lVert a-x\rVert_1} = \sum_{(y_1,y_2)\in \mathbb{Z}^2} \mathrm{e}^{-\beta |y_1|-\beta |y_2|} = \left(\frac{\mathrm{e}^\beta+1}{\mathrm{e}^\beta - 1 }\right)^2 \notag \\
    &=:\kappa_\beta. \label{eq:kappa_beta}
\end{align}
As such, in a ``distributional sense" we have $\langle 1|K_\beta = \kappa_\beta \langle 1|$.  Since $A$ is bijective, we also have $K_\beta |1\rangle = \kappa_\beta |1\rangle$ by an analogous argument.  Hence,   \begin{equation}
    \langle 1| K_\beta^n|a\rangle = \langle a|K_\beta^n|1\rangle = \kappa_\beta^n <\infty,
\end{equation}
along with $\langle b_+|K_\beta^n|b_-\rangle \le \langle b_+|K_\beta^n|1\rangle <\infty$, which verifies the existence of all Gibbs measures.

\subsection{Linear confinement and soundness}
\begin{defn}[Adapted 1-norm]
    Given the eigenvectors $u$ and $v$ from Definition \ref{def:eigenvectors}, define the \textbf{adapted 1-norm}  for any $x\in \mathbb{R}^2$: 
    \begin{equation}
        \lVert x\rVert_{1,A} := | x\cdot u| + |x\cdot v|.
    \end{equation}
\end{defn}
\begin{prop}[Equivalence of norms]\label{prop:equivalence-norms}
    The following inequalities hold:
    \begin{equation} \label{eq:equivalence-norms}
        \frac{1}{\sqrt{2}} \lVert x\rVert_1 \le  \lVert x\rVert_{1,A} \le \sqrt{2} \lVert x\rVert_1.
    \end{equation}
\end{prop}
\begin{proof}
    Since $A$ is symmetric, $u$ and $v$ are an orthonormal basis of $\mathbb{R}^2$.  As such the standard bounds
    \begin{equation}
        \lVert x\rVert_2 \le \lVert x\rVert_1\le \sqrt{2} \lVert x\rVert_2
    \end{equation}hold for $\lVert x\rVert_{1,A}$ too.  So we have $\lVert x\rVert_{1,A} \le \sqrt{2}\lVert x\rVert_2 \le \sqrt{2}\lVert x\rVert_1$; we can switch $\lVert x\rVert_{1,A}$ and $\lVert x\rVert_1$ as well.  This gives \eqref{eq:equivalence-norms}. 
\end{proof}

\begin{defn}[Ideal solution/region]\label{def:ideal}
    For any $c\in \mathbb{R}^2$, define the \textbf{ideal solution} as the set of vectors $\lbrace \xi_j\rbrace \in (\mathbb{R}^2)^{\mathbb{Z}}$ defined by \begin{equation}
        \xi_j := A^{-j}c.
    \end{equation}
    By inspection, we have $h(\xi_j,\xi_{j+1}) = 0$.  We further define \textbf{displaced coordinates}\begin{equation}
        d_i := x_i - \xi_i =: s_i u + t_i v,
    \end{equation}
    with $d_i \in -\xi_i + \mathbb{Z}^2$, and $(s_i, t_i)\in \mathbb{R}^2$.  Lastly, fix an integer $M>\mathcal{R}_{\text{tol}}$, and define the \textbf{ideal region} $\mathcal{I}_j \subset \mathbb{Z}^2$, and its complement \textbf{non-ideal region} $\mathcal{N}_j := \mathbb{Z}^2\setminus \mathcal{I}_j$, via \begin{equation}
        \mathcal{I}_j := \lbrace x \in \mathbb{Z}^2 : \lVert x - \xi_j \rVert_1 \le M \rbrace. \label{eq:Ij_def}
    \end{equation}
\end{defn}

\begin{thm}[Linear soundness]\label{thm:sound}
    Let $\Lambda = \lbrace j_-, \ldots, j_+\rbrace \subset \mathbb{Z}$ be a finite connected interval.  Pick endpoints $b_-\in \mathcal{I}_{j_-}$ and $b_+\in \mathcal{I}_{j_+}$. For $x\in (\mathbb{Z}^2)^{\Lambda}$, define the energy \begin{equation}
        E(x) := H_{\Lambda,b_-,b_+}(x)
    \end{equation}
    and the linear displacement \begin{equation}\label{eq:Dx-def}
        D(x) := \sum_{i=j_-}^{j_+}\lVert x_i - \xi_i\rVert_1.
    \end{equation}
    Then there exists constant $C$ depending only on $\lambda$ and $M$, and constants $\eta$ and $\eta_+$ depending only on $\lambda$,  such that the following \textbf{linear soundness} inequalities hold: \begin{equation} \label{eq:soundness}
        \eta D(x) - C\le E(x) \le \eta_+ D(x).
    \end{equation}
\end{thm}
\begin{proof}
    Notice that $d_{j_-} = b_--\xi_{j_-}$ is bounded: \begin{equation}
        \lVert d_{j_-}\rVert_{1,A } \le \sqrt{2}M;
    \end{equation}a similar estimate holds for $d_{j_+}$. We first study the energy using the adapted energy: \begin{align}
        E_A(x) &:= \sum_{i=j_-}^{j_+-1} \lVert x_i - Ax_{i+1}\rVert_{1,A} =\sum_{i=j_-}^{j_+-1} \lVert d_i - Ad_{i+1}\rVert_{1,A} \notag \\
        &=\sum_{i=j_-}^{j_+-1} \left( | s_i - \lambda s_{i+1}| + |t_i - \lambda^{-1}t_{i+1}| \right) . \label{eq:EA_def}
    \end{align}
    Apply the triangle inequality and other simple inequalities, along with \eqref{eq:Ij_def}, to obtain: \begin{subequations} \label{eq:thmsoundinequalities}
        \begin{align}
            E_A(x) &\ge (\lambda-1)\sum_{i= j_- +1 }^{j_+-1 }|s_i| - \lambda |s_{j_+}|-|s_{j_-}| + \frac{\lambda-1}{\lambda}\sum_{i=j_-+1}^{j_+-1}|t_i| - \frac{1}{\lambda}|t_{j_+}| - |t_{j_-}| \notag \\
            &\ge \frac{\lambda-1}{\lambda} D_A(x) - 2\left[\lambda - 1+ \sqrt{2}\left(\lambda + \frac{1}{\lambda}\right)\right] M \label{eq:thmsoundinequalities-1} \\
            E_A(x) &\le (\lambda+1)\sum_{i= j_- +1}^{j_+-1}|s_i| + \lambda |s_{j_+}|+|s_{j_-}| + \frac{\lambda+1}{\lambda}\sum_{i=j_-+1}^{j_+-1}|t_i| + \frac{1}{\lambda}|t_{j_+}| + |t_{j_-}| \notag \\
            &\le (\lambda+1) D_A(x)   \label{eq:thmsoundinequalities-2}
        \end{align}
    \end{subequations}
    where $D_A(x)$ is defined analogously to \eqref{eq:Dx-def} but with the adapted norm. We can apply Proposition \ref{prop:equivalence-norms} to \eqref{eq:EA_def} to find \begin{equation}
      \sqrt{2} E_A(x) \ge   E(x) \ge \frac{1}{\sqrt{2}} E_A(x).
    \end{equation}Applying Proposition \ref{prop:equivalence-norms} again to compare $D(x)$ with $D_A(x)$ analogously, and defining \begin{subequations}
        \begin{align}
            \eta &:= \frac{\lambda -1}{2\lambda}, \\
            \eta_+ &:= 2(\lambda+1), \\
            C &:= 2\sqrt{2}\left[\lambda - 1+ \sqrt{2}\left(\lambda + \frac{1}{\lambda}\right)\right] M.
        \end{align}
    \end{subequations}
     we find the lower bound in \eqref{eq:soundness} from \eqref{eq:thmsoundinequalities-1}, and the upper bound from \eqref{eq:thmsoundinequalities-2}.
\end{proof}
\subsection{No large excursions from the ideal solution}\label{sec:excursion}
Theorem \ref{thm:sound} very strongly penalizes large deviations from the ideal solution in Definition \ref{def:ideal}.  More concretely, we will show that configurations that stray farther than $M$ from the ideal solution are heavily suppressed in the Gibbs state $\mu_{\Lambda,\beta,\mathcal{R}(\xi_{j_-}),\mathcal{R}(\xi_{j_+})}$. The idea will be to prove that most walks in the partition function $\langle b_+|K_\beta^{j_+-j_-}|b_-\rangle$ stay close to $\mathcal{I}_j$.  To count the contribution of walks that avoid $\mathcal{I}_j$ (at each $j$) in the partition function, we define:
\begin{defn}[Ideal/non-ideal transfer kernel]
    Given integers $j$ and $n>1$, $b_-\in \mathcal{I}_{j}$ and $b_+\in \mathcal{I}_{j+n}$, define the finite-dimensional matrices $F_{j+n,j}, G_{j+n,j}:\mathbb{R}^{\mathcal{I}_j} \rightarrow \mathbb{R}^{\mathcal{I}_{j+n}}$ by \begin{subequations}\begin{align}
    \langle b_+| G_{j+n, j}|b_-\rangle &= \sum_{x_{n-1}\in \mathcal{I}_{j+n-1}}\cdots \sum_{x_1\in\mathcal{I}_{j+1}} \langle b_+|K_\beta | x_{n-1} \rangle  \langle x_{n-1}|K_\beta |x_{n-2}\rangle \cdots \langle x_1|K_\beta |b_-\rangle, \\
        \langle b_+| F_{j+n, j}|b_-\rangle &= \sum_{x_{n-1}\in \mathcal{N}_{j+n-1}}\cdots \sum_{x_1\in\mathcal{N}_{j+1}} \langle b_+|K_\beta | x_{n-1} \rangle  \langle x_{n-1}|K_\beta |x_{n-2}\rangle \cdots \langle x_1|K_\beta |b_-\rangle.
    \end{align}\end{subequations}
    $G_{j+n,j}$ is the ideal transfer kernel; $F_{j+n,j}$ is the non-ideal kernel.  Define $G_{j,j}$ to be the identity matrix, while $\langle b_+| G_{j+1,j} |b_-\rangle = \langle b_+| K_\beta |b_-\rangle$.
\end{defn}
\noindent It is straightforward to establish the following useful representation of the finite-volume partition function by classifying all possible paths between $b_-$ and $b_+$ by whether they traverse $\mathcal{N}_j$ or $\mathcal{I}_j$ at each step: 
\begin{prop}[Transfer kernel representation of the partition function] \label{prop:transferkernel}
    The set of possible places that a path can traverse the non-ideal sets between $j_-$ and $j_+$ is characterized by a sequence of pairs of integers $\lbrace (k_m, \ell_m)\rbrace \in \mathcal{V}$, where \begin{equation}
        \mathcal{V} := \lbrace 0\rbrace  \cup \lbrace (k_i,\ell_i) \text{ for } 1\le i \le m : m\ge1, \; k_i+1<\ell_i, \;  \ell_i \le k_{i+1}, \; j_- \le k_1 ,  \; \ell_m \le j_+\rbrace,
    \end{equation}and $m$ is allowed to be as large as possible, subject to all of the above constraints. We then have   \begin{align}
        \langle b_+|K_\beta^{j_+-j_-}|b_-\rangle &= \langle b_+| G_{j_+,j_-}|b_-\rangle \notag \\
        &\;\;\; + \langle b_+|\sum_{\lbrace (k_m, \ell_m)\rbrace \in \mathcal{V}\setminus\lbrace 0\rbrace} G_{j_+,\ell_m}F_{\ell_m,k_m}\cdots G_{k_2,\ell_1}F_{\ell_1,k_1}G_{k_1,j_-} |b_-\rangle,
    \end{align}
    where $0 \in \mathcal{V}$ corresponds to the first term above.
\end{prop}

\begin{lem}[Exponential suppression of large excursions]
\label{lem:excursion}
Let $\lbrace j,\ldots, j+n\rbrace \subset \Lambda$.  For sufficiently large $M$, there exist constants $B>0$ and $0<\theta<1$ such that\begin{equation}
    \mathbb{P}_{j+n,j}^{\text{excursion}} := \mu_{\beta,\Lambda,b_-,b_+}(x_{i}\in \mathcal{N}_i \text{ if } j<i<j+n, x_{j+n}\in \mathcal{I}_{j+n}|x_j = c_- \in \mathcal{I}_j) < B \theta^{n-1}. \label{eq:excursion}
\end{equation}
\end{lem}
\begin{proof}
    Let us first show that the fraction of weight in the Gibbs state, contained in the paths which begin at $c_-\in \mathcal{I}_j$ at $j$ and next hit the ideal region at site $c_+ \in \mathcal{I}_{j+n}$ after $n>1$ steps to the right is exponentially small for sufficiently large $n$, for suitable choices of $M$.   Define \begin{equation}
        \sigma :=\sup_{\xi \in \mathbb{R}^2/\mathbb{Z}^2} \sum_{x \in \mathbb{Z}^2} \mathbb{I}\left(\lVert x - \xi \rVert_1 >M\right) \mathrm{e}^{-\beta  \eta \lVert x - \xi \rVert_1},
    \end{equation}
    and notice the elementary inequality \begin{equation} \label{eq:sigma_inequality}
        \sigma < \sum_{x \in \mathbb{Z}^2}\mathbb{I}(\lVert x\rVert_1 \ge M-1) \mathrm{e}^{-\beta \eta  \lVert x\rVert_1} = 4\sum_{j=M-1}^\infty j\mathrm{e}^{-\beta \eta  j} = 4 \frac{(M-1) - (M-2)\mathrm{e}^{-\beta \eta }}{(1-\mathrm{e}^{-\beta \eta})^2}\mathrm{e}^{-\beta \eta  (M-1)}.
    \end{equation} 
    Importantly we can take $\sigma$ as small as we wish by making $M$ large enough.  
    
    Now, by Theorem \ref{thm:sound}, we know that \begin{align}
        \max_{c_+\in \mathcal{I}_{j+n}, \; c_-\in \mathcal{I}_j} \langle c_+| F_{j+n,j}|c_-\rangle &\le \mathrm{e}^{\beta C}\sum_{x_{j+1} \in \mathcal{N}_{j+1}, \ldots, x_{j+n-1}\in\mathcal{N}_{j+n-1}} \prod_{i=1}^{n-1} \mathrm{e}^{-\beta \eta \lVert x_{j+i}- \xi_{j+i}\rVert_1} \notag \\
        &\le \mathrm{e}^{\beta C}\sigma^{n-1}. \label{eq:max-of-F}
    \end{align}
    In contrast, we can get a very crude lower bound on the weight of paths that stay in the ideal region by counting only the weight of a single path which follows the rounded map as much as possible: \begin{align}
        \min_{c_+\in \mathcal{I}_{j+n}, \; c_-\in \mathcal{I}_j} \langle c_+| G_{j+n,j}|c_-\rangle &> \langle c_+|K_\beta | \mathcal{R}(\xi_{j+n-1}) \rangle \langle \mathcal{R}(\xi_{j+n-1})|K_\beta|\mathcal{R}(\xi_{j+n-2})\rangle \cdots \langle \mathcal{R}(\xi_{j+1})|K_\beta|c_-\rangle \notag \\
        &> \mathrm{e}^{-\beta (1+\lVert A\rVert_{1\rightarrow 1}) \mathcal{R}_{\text{tol}} (n-1)} \mathrm{e}^{-2\beta (1+\lVert A\rVert_{1\rightarrow 1}) (M+\mathcal{R}_{\text{tol}}) } .
    \end{align}
    where the factors in the second line come from bounding the matrix elements along the rounded ideal solution, along with the worst possible boundary terms, respectively.  Defining \begin{subequations} \label{eq:zeta-Q-def}
        \begin{align}
            \zeta &:= \mathrm{e}^{-\beta (1+\lVert A\rVert_{1\rightarrow 1})\mathcal{R}_{\text{tol}}}, \\
            Q &:= \mathrm{e}^{-2\beta (1+\lVert A\rVert_{1\rightarrow 1})(M+\mathcal{R}_{\text{tol}})},
        \end{align}
    \end{subequations}
    we deduce that \begin{align} \label{eq:P-excursion}
       \mathbb{P}_{j+n,j}^{\text{excursion}} &= \dfrac{\displaystyle \sum_{c_+\in \mathcal{I}_{j+n}}\langle b_+|K_\beta^{j_+-j-n}|c_+\rangle \langle c_+|F_{j+n,j}|c_-\rangle  }{\langle b_+|K_\beta^{j_+-j}|c_-\rangle } <  \dfrac{\displaystyle \sum_{c_+\in \mathcal{I}_{j+n}}\langle b_+|K_\beta^{j_+-j-n}|c_+\rangle \langle c_+|F_{j+n,j}|c_-\rangle  }{\displaystyle \sum_{c_+\in \mathcal{I}_{j+n}}\langle b_+|K_\beta^{j_+-j-n}|c_+\rangle \langle c_+|G_{j+n,j}|c_-\rangle } \notag \\
       &< \max_{c_+\in \mathcal{I}_{j+n}, \; c_-\in \mathcal{I}_j} \frac{\langle c_+| F_{j+n,j}|c_-\rangle}{\langle c_+| G_{j+n,j}|c_-\rangle} < \frac{\mathrm{e}^{\beta C}}{Q} \left(\frac{\sigma}{\zeta}\right)^{n-1}.
    \end{align}
    Choose $M$ so that $\sigma<\zeta$; then define $\theta := \sigma/\zeta$.  To obtain \eqref{eq:excursion}, we note that we may take $B := \mathrm{e}^{\beta C}/Q$.
\end{proof}

\section{Constructing the DLR state} \label{sec:dlr}
We are now ready to prove the main theorem.  The key strategy is to use Lemma \ref{lem:excursion} to show that the Gibbs state does not stray too far from the ideal solution.  To make this precise, we will introduce an inhomogeneous Markov chain, first reading the Gibbs state Markov field from left to right, and use path coupling arguments \cite{BubleyDyer1997PathCoupling} to prove that it ``mixes" and forgets its left initial conditions when conditioned on a fixed right endpoint.   Of course, we can just as easily read from right to left, a fact we will eventually exploit in Section \ref{sec:thmproof}.

\subsection{Constructing an auxiliary inhomogeneous Markov chain}
 Analogous to Definition \ref{def:ideal}, we introduce: \begin{defn}[Rounded displaced coordinates] Define \begin{equation}
    \widetilde{d}_i := x_i - \mathcal{R}(\xi_i), \label{eq:rounded_displaced_coordinates}
\end{equation}
along with the set \begin{equation}
    \widetilde{\mathcal{I}} := \lbrace \widetilde{d}\in \mathbb{Z}^2 : \widetilde{d} + \mathcal{R}(\xi_j) \in \mathcal{I}_j \text{ for some } j\rbrace.
\end{equation}
\end{defn}
\noindent Notice that $\widetilde{\mathcal{I}}$ is a bounded set, e.g. if $x\in \widetilde{\mathcal{I}}$ then a crude bound gives \begin{equation}
 \lVert x\rVert_1 < M+ \mathcal{R}_{\text{tol}}+1   . \label{eq:xtildeIcrude}
\end{equation}

\begin{defn}[Auxiliary inhomogeneous Markov chain]\label{def:auxiliary}
    Index ``time" by $j\in \mathbb{Z}$, and introduce the following inhomogeneous Markov chain on the state space $\widetilde{\mathcal{I}}\times \mathbb{Z}_{\ge 0}$:
    \begin{subequations}\begin{align}
        \langle x^\prime,n^\prime  |P_j|x,n>0\rangle  &= \delta_{x,x^\prime} \delta_{n-1,n^\prime}, \\
        \langle x^\prime,0 |P_j|x,0\rangle &= \frac{\langle b_+| K_\beta^{j_+-j-1}|x^\prime + \mathcal{R}(\xi_{j+1})\rangle \langle x^\prime +\mathcal{R}(\xi_{j+1})|G_{j+1,j}|x + \mathcal{R}(\xi_j)\rangle  }{\langle b_+| K_\beta^{j_+-j}|x + \mathcal{R}(\xi_j)\rangle } , \label{eq:k=0Pj} \\
        \langle x^\prime,n^\prime -1 |P_j|x,0\rangle &= \frac{\langle b_+| K_\beta^{j_+-j-n^\prime}|x^\prime + \mathcal{R}(\xi_{j+n^\prime})\rangle \langle x^\prime +\mathcal{R}(\xi_{j+n^\prime})|F_{j+n^\prime,j}|x + \mathcal{R}(\xi_j)\rangle  }{\langle b_+| K_\beta^{j_+-j}|x + \mathcal{R}(\xi_j)\rangle }, \label{eq:nprime>1}
    \end{align}\end{subequations}where $j_+-j\ge n^\prime >1$ in \eqref{eq:nprime>1}, and $P_j(X,X^\prime)$ is the probability that at time $j$ we move from state $X$ to $X^\prime$.   To avoid overloading the notation, we suppress the dependence of $P$ on $b_+$.
\end{defn}
\noindent Proposition \ref{prop:transferkernel} ensures that $P_j$ defines a genuine inhomogeneous Markov chain.   We should emphasize the physical interpretation of the chain:  the space $\widetilde{\mathcal{I}}$ is indexing all possible values of $\widetilde{d}_i$ coordinate at a given ``time" (location in the chain).  For ease of notation we will usually annotate it as $x$ moving forward.  The reader should keep in mind that at ``time" $j$, only $x\in \mathcal{I}_j - \mathcal{R}(\xi_j)$ are allowed as coordinates for this auxiliary Markov chain.

\begin{prop} \label{prop:0element-bound}
    For every $x,x^\prime \in \widetilde{\mathcal{I}}$, if $x + \mathcal{R}(\xi_j)\in \mathcal{I}_j$ and  $x^\prime + \mathcal{R}(\xi_{j+1})\in \mathcal{I}_{j+1}$, \begin{equation}
        \langle x^\prime,0 |P_j|x,0\rangle > \epsilon_0 \label{eq:epsilon0}
    \end{equation}
    for an $\epsilon_0>0$ independent of the size of $\Lambda$.
\end{prop}
\begin{proof}
    We find simple (and far from tight) upper and lower bounds on the (relative) fluctuations of the matrix elements in \eqref{eq:k=0Pj}.  First, we use the triangle inequality: for any $x,y,z\in \mathbb{Z}^2$, \begin{equation} \label{eq:harnack}
        \mathrm{e}^{-\beta \lVert y-z\rVert_1 }<\frac{\langle x|K_\beta |y\rangle}{\langle x|K_\beta |z\rangle }= \mathrm{e}^{-\beta \left( \lVert Ax-y\rVert_1 - \lVert Ax -z \rVert_1 \right)} < \mathrm{e}^{\beta \lVert y-z\rVert_1 }.
    \end{equation}
    Next, consider the following ``first return" identity, analogous to Proposition \ref{prop:transferkernel}: \begin{equation}
        \langle b_+|K^{j_+-j}_\beta|x\rangle = \sum_{y\in \mathcal{I}_{j+1}}\langle b_+|K_\beta^{j_+-j-1}|y\rangle \langle y|K_\beta |x\rangle + \sum_{n= 2}^{j_+-j}\sum_{y\in \mathcal{I}_{j+n}}\langle b_+|K_\beta^{j_+-j-n}|y\rangle\langle y|F_{j+n,j}|x\rangle.
    \end{equation}
    Using manipulations similar to \eqref{eq:P-excursion}, for any $n>1$, $y\in \mathcal{I}_{j+1}$ and $y\in \mathcal{I}_{j+1}$,  \begin{align}
         \langle b_+|K_\beta^{j_+-j-1}|y\rangle &> \langle b_+|K_\beta^{j_+-j-n}|\mathcal{R}(\xi_{j+n})\rangle \prod_{i=2}^{n-1}\langle \mathcal{R}(\xi_{j+i+1})|K_\beta|\mathcal{R}(\xi_{j+i})\rangle \cdot \langle \mathcal{R}(\xi_{j+2})|K_\beta|x\rangle \notag \\
         &> \langle b_+|K_\beta^{j_+-j-n}|\mathcal{R}(\xi_{j+n})\rangle \cdot Q \zeta^{n-2}
    \end{align}
where $\zeta$ and $Q$ were defined in \eqref{eq:zeta-Q-def}.  Now using \eqref{eq:max-of-F} along with \eqref{eq:harnack}, for any $y\in\mathcal{I}_{j+1}$, \begin{equation}
    \langle b_+|K^{j_+-j}_\beta|x\rangle <   |\widetilde{\mathcal{I}}|  \langle b_+|K_\beta^{j_+-j-1}|y\rangle \left[\frac{1}{Q} + \frac{\mathrm{e}^{\beta C}}{Q}\sum_{n=2}^\infty \left(\frac{\sigma}{\zeta}\right)^{n-2}\right] =: \mathcal{M} \langle b_+|K_\beta^{j_+-j-1}|y\rangle
\end{equation}
for some constant $\mathcal{M}$ independent of $j_+$.  Combining this inequality with \eqref{eq:k=0Pj}, we get the crude bound \begin{equation}
    \epsilon_0 := \frac{Q}{\mathcal{M}} >0
\end{equation}
which completes the proof.
\end{proof}

Define the auxiliary function $V:\widetilde{\mathcal{I}}\times \mathbb{Z}_{\ge 0} \rightarrow \mathbb{R}_+$ as (via the bra-ket notation) \begin{equation}
    \langle V| := \sum_{x\in \widetilde{\mathcal{I}}}\sum_{n=0}^\infty \langle x,n| r^n
\end{equation}for some coefficient \begin{equation}
    1<r<\theta^{-1}, \label{eq:rthetabound}
\end{equation}
which we will pick later.  
\begin{prop}\label{prop:bounded_V}
    Given any measure $ |\mu\rangle $ on $\widetilde{\mathcal{I}}\times \mathbb{Z}_{\ge 0}$, if $\langle V|\mu\rangle<\infty$, then $\langle V|P_j|\mu\rangle <\infty$.
\end{prop}
\begin{proof}
    Decompose normalized measure $|\mu\rangle = |\mu_0\rangle + |\mu_>\rangle$, where $|\mu_0\rangle$ contains terms with $n=0$ and $|\mu_>\rangle$ has $n>0$ respectively.  Notice that \begin{subequations}\begin{align}
        \langle V|P_j|\mu_>\rangle &=  r^{-1} \langle V|\mu_>\rangle, \label{eq:V>mu} \\
        \langle V|P_j|\mu_0\rangle &= \sum_{n=0}^\infty r^n \sum_{x\in\widetilde{\mathcal{I}}} \langle x ,n|P_j|\mu_0\rangle < 1+\sum_{n=1}^\infty r^n B \theta^{n} = 1+\frac{Br\theta }{1-r\theta} =: J. \label{eq:V0mu}
    \end{align}\end{subequations}
    where \eqref{eq:V>mu} comes from the fact that $P_j$ deterministically drops $n \rightarrow n-1$ if $n>0$, and \eqref{eq:V0mu} comes from \eqref{eq:excursion}.    Therefore for any $j$, \begin{equation}
        \langle V|P_j|\mu\rangle \le \frac{\langle V|\mu\rangle}{r}+J < \infty, \label{eq:Vexpbound}
    \end{equation}
    which completes the proof.
\end{proof}

\subsection{Path coupling}
\begin{defn}[Path coupling with flag]\label{def:pathcoupling}
    Let $\pi(X,Y,f)$ be a measure on $(\widetilde{\mathcal{I}}\times \mathbb{Z}_{\ge 0})^2 \times \lbrace 0,*\rbrace$, with $X,Y\in \widetilde{\mathcal{I}}\times \mathbb{Z}_{\ge0}$ and $f\in \lbrace 0, *\rbrace$, with \begin{subequations}
        \begin{align}
            \sum_{Y,f} \pi(X,Y,f) &= \mu(X), \\
            \sum_{X,f} \pi(X,Y,f) &= \nu(Y),
        \end{align}
    \end{subequations}
    where $f\in \lbrace 0 , *\rbrace$ is a flag variable that asserts that the chains have coupled, and we demand that if $X\ne Y$, \begin{equation}
        \pi(X,Y,*) = 0.
    \end{equation}  We say that $\pi(X,Y,f) \in \text{couple}(\mu,\nu)$ is a valid \textbf{coupling}.  Given a Markov chain on a single copy, we lift the transition rules to the doubled space such that \begin{subequations}\begin{align}
       \sum_{Y^\prime,f^\prime} \langle X^\prime Y^\prime f^\prime  |\mathcal{P}_j|XYf\rangle &=  \langle X^\prime |P_j|X\rangle,    \\
        \sum_{X^\prime, f^\prime} \langle X^\prime Y^\prime f^\prime  |\mathcal{P}_j|XYf\rangle & =\langle Y^\prime |P_j|Y\rangle.
    \end{align}\end{subequations}
    We say that a coupling is \textbf{suitable} if it has the following property:  for some $\epsilon>0$, \begin{equation}
       \sum_{x^{\prime\prime}\in\widetilde{\mathcal{I}}}\langle (x^{\prime\prime},0)(x^{\prime\prime},0),*| \mathcal{P}_j|(x,0)(x^\prime,0),0\rangle  > \epsilon 
    \end{equation}for any $x,x^\prime$, and for all $X,X^\prime,X^{\prime\prime}\in \widetilde{\mathcal{I}}\times \mathbb{Z}_{\ge 0}$, \begin{equation}
        \langle X^\prime X^{\prime\prime}f|\mathcal{P}_j|XX*\rangle = \delta_{X^\prime,X^{\prime\prime}} \delta_{f,*}\langle X^\prime|P_j|X\rangle , \label{eq:coupletogether}
    \end{equation}
    The last equation ensures that once the chains have coupled, they never uncouple.  Hence this is a standard path coupling of two Markov chains \cite{BubleyDyer1997PathCoupling}.
\end{defn}
\begin{prop}
    A suitable coupling exists with \begin{equation}
        \epsilon \ge \epsilon_0^2.
    \end{equation}
\end{prop}  
\begin{proof}
    Consider a coupling where we choose the two coupled chains to evolve independently if $f=0$, while they evolve together if $f=*$, so \eqref{eq:coupletogether} holds.  More concretely, we choose transition rules such that: \begin{enumerate}
        \item When $f=0$, the two chains evolve independently; when $f=*$, the two chains evolve identically as in \eqref{eq:coupletogether}.
        \item When $f=0$, if the two chains (independently) are sent to the same state $(x,0)$ in the ideal region, the flag $f$ is subsequently set to $*$.  Otherwise, maintain $f=0$.
    \end{enumerate} 
    The important thing to keep in mind is that: \begin{equation}\label{eq:flip-flag}
        \langle (x,0)(x,0)0|\mathcal{P}_j  =0,
    \end{equation}
    and we always coalesce the two chains in the ideal region if possible.  With this in mind, we can use Proposition \ref{prop:0element-bound} to see that  $\langle (0,0)(0,0)*| P_j\otimes P_j|(x,0)(x^\prime,0)0\rangle \ge \epsilon_0^2$.  
\end{proof} 

The key step in the proof of the theorem is to show that using a suitable coupling chain $\mathcal{P}_j$,  the total variation distance between any two (bounded $V$) measures exponentially shrinks with $j$.  To establish this, we follow a technique described in  \cite{Hairer:2011harris}. \begin{defn}[Optimal coupling cost] \label{def:optimalcost}
    For $\gamma>0$, define the \textbf{optimal coupling cost} \begin{equation}
        W(\mu,\nu) := \inf_{\pi \in \text{couple}(\mu,\nu)} \langle W|\pi\rangle
    \end{equation}
    where \begin{equation}
        \langle W| := \sum_{X,X^\prime} \langle XX^\prime 0|  \left[ 1  + \gamma V(X) +\gamma  V(X^\prime) \right]. \label{eq:W_vec}
    \end{equation}
\end{defn} 
\begin{prop}\label{prop:TV}
    $W$ bounds the total variation distance: \begin{equation}
        W(\mu,\nu) \ge \lVert \mu - \nu \rVert_{\text{TV}} := \frac{1}{2}\sum_{X} |\mu(X)-\nu(X)|. \label{eq:W_TV}
    \end{equation}
\end{prop}
\begin{proof}
    $W$ is an increasing function of $\gamma$, so it is lower bounded by its value at $\gamma=0$.  Hence $W(\pi) \ge  \pi(f=0)$.  Clearly, \begin{equation}
       1-\pi(f=0) =  \pi(f=*) \le \sum_{X}\min(\mu(X),\nu(X)).
    \end{equation} A textbook manipulation then shows \eqref{eq:W_TV}. 
\end{proof}

The goal is now to show that for arbitrary suitable coupling and initial conditions, $W(\pi)$ decreases exponentially at an $N$-independent rate.  To show this, we will work with the $L$-step Markov transition matrix \begin{equation}
    \mathcal{Q}_j := \mathcal{P}_{j+L-1}\mathcal{P}_{j+L-2}\cdots \mathcal{P}_j,
\end{equation}
where we define $Q_j := P_{j+L-1}\cdots P_j$ analogously.  The integer $L$ is chosen such that the following inequality holds: \begin{equation}
    r^L>  \frac{2 J}{(1-r^{-1})^2}. \label{eq:define_L}
\end{equation}
\begin{lem}\label{lem:W_contraction}
    For sufficiently small $\gamma$, there exists $0<\rho<1$ such that if $\langle V|\mu\rangle$ and $\langle V|\nu\rangle$ are finite,  \begin{equation}
        W(Q_j |\mu\rangle, Q_j |\nu\rangle) \le \rho \cdot W(|\mu\rangle, |\nu\rangle ).  \label{eq:W_contraction}
    \end{equation}
\end{lem}
\begin{proof}
    It suffices to consider the case $|\mu\rangle = |X\rangle$ and $|\nu \rangle = |X^\prime\rangle$ and present a suitable coupling $\mathcal{Q}_j$ that obeys \begin{equation}
        \langle W|\mathcal{Q}_j|XX^\prime 0 \rangle \le \rho \langle W|XX^\prime 0\rangle.  \label{eq:WXXprime}
    \end{equation}  By Definition \ref{def:optimalcost}, $W$ is defined as an infimum over all path couplings, so the infimum will also obey \eqref{eq:W_contraction}; by linearity of $\langle W|\pi\rangle$, \eqref{eq:WXXprime} extends to $\langle W|\mathcal{Q}_j|\pi\rangle \le \rho \langle W|\pi\rangle$ for any coupling $\pi$.  In what follows, take $X\ne X^\prime$, because the minimal coupling has $f=*$ when $X=X^\prime$.

    First, we claim that
    \begin{equation}
        \langle W| \mathcal{Q}_j|XX^\prime\rangle \le 1 + \gamma \langle V|Q_j|X\rangle +  \gamma \langle V|Q_j|X^\prime \rangle. \label{eq:W_separated}
    \end{equation} 
    The inequality is saturated if and only if $\mathbb{P}[f=*]=0$ -- the chains have not coupled with certainty.  With this in mind, let us apply repeatedly \eqref{eq:Vexpbound} to bound $\langle V|Q_j|X\rangle$: \begin{align}
        \langle V|Q_j|(x,k)\rangle \le \left\lbrace \begin{array}{ll} r^{k-L} &\ k \ge L \\ (1-r^{-1})^{-1}J &\ k < L \end{array}\right.,
    \end{align}
    or uniformly, \begin{equation}
        \langle V|Q_j|\mu\rangle \le r^{-L}\langle V|\mu\rangle + (1-r^{-1})^{-1}J.
    \end{equation}

    Now, suppose that $\langle W|XX^\prime0\rangle \ge  1+\gamma (1+r^{L})$ with $L$ defined in \eqref{eq:define_L}. Then we find \begin{align}
        \langle W| \mathcal{Q}_j|XX^\prime0\rangle &\le 1 + \gamma r^{-L}\frac{\langle W|XX^\prime\rangle -1}{\gamma} + \frac{2\gamma J}{1-r^{-1}} \notag \\
        &\le \rho_{\text{far}}\langle W|XX^\prime0\rangle \label{eq:rhofar-used}
    \end{align}
    where we have employed \eqref{eq:W_separated} and defined \begin{equation}
        \rho_{\text{far}} := r^{-L} + \dfrac{\displaystyle 1 - r^{-L} + \dfrac{2\gamma J}{1-r^{-1}}}{\displaystyle 1 + \gamma (1+r^L)}.
    \end{equation} Notice that $\rho_{\text{far}}<1$ since, due to \eqref{eq:define_L}, \begin{equation}
        1-\rho_{\text{far}} = \gamma  \dfrac{\displaystyle  r^L-r^{-L}-\dfrac{2 J}{1-r^{-1}} }{\displaystyle 1 + \gamma (1+r^L)} > \gamma  \dfrac{\displaystyle  r^{L-1}-r^{-L}}{\displaystyle 1 + \gamma (1+r^L)} >  0.
    \end{equation} We note that \eqref{eq:rhofar-used} is an inequality because the chain could in principle couple the two chains at a state $\notin \widetilde{\mathcal{I}}$, and we have neglected this possibility in our estimate. 

    Now suppose that $\langle W|XX^\prime0\rangle < 1 + \gamma (1+r^{L})$.  Then we must have (in principle) both chains hit the ideal ($n=0$) region at some point during the interval from $j$ to $j+L-1$: if $\langle W|XX^\prime\rangle = 1 + \gamma(1+r^L)$ it means that one chain is at most $L$ steps away from $n=0$ (if the other already has $n=0$).  Since we see a smaller value than this, we actually know that both chains must individually reach $n=0$ for at least one time in $\lbrace j,\ldots, j+L-1\rbrace$.    Because the chain is suitable, there is then a finite probability that the chains will coalesce during these $L$ steps.  A very crude lower bound on this probability is to evaluate the probability of the following event:  after each chain hits $n=0$, they each follow the rounded solution $(0,0)$ until they coalesce?  If chain $X$ hits $n=0$ at time $i$ and $X^\prime$ hits at time $i^\prime$,  using Proposition \ref{prop:0element-bound} we find: \begin{equation}
        \mathbb{P}[\text{this coalescence event}] \ge \epsilon_0^{|i^\prime-i|+2}\ge \epsilon_0^{L+1}. 
    \end{equation}
    If this event happens, we lose weight $\ge 1$ in $W$ because, by \eqref{eq:flip-flag}, the flag variable becomes $*$.  Hence we find: \begin{align}
        \langle W|\mathcal{Q}_j|XX^\prime 0 \rangle &< -\epsilon_0^{L+1} + 1 +  \frac{2\gamma J}{1-r^{-1}} + r^{-L}\left(\langle W|XX^\prime\rangle -1 \right) \notag \\
        &< 1-\epsilon_0^{L+1} +  \frac{2\gamma J}{1-r^{-1}} + \gamma \left(1+r^{-L} \right) =: \rho_{\text{near}} .
    \end{align}
    Choose $\gamma$ sufficiently small so that $\rho_{\text{near}}<1$ as well, and take \begin{equation}
        \rho = \max(\rho_{\text{near}},\rho_{\text{far}}) < 1.
    \end{equation}
    Since this was not an optimal coupling, the optimal coupling must obey \eqref{eq:W_contraction}.
\end{proof}
\begin{cor}\label{cor:left_boundary}
    For any $j_-<p<q<j_+$, let $\mu^{\lbrace p, q\rbrace}_{\Lambda,\beta, b_-,b_+}$ denote the Gibbs state restricted to the interval $\lbrace p, \ldots, q\rbrace$, with $b_\pm\in \mathcal{I}_{j_\pm}$.   Consider changing the domain of interest to $\Lambda^\prime = \lbrace j_-^\prime ,\ldots ,j_+\rbrace$ with $j_-^\prime < j_-$ and new boundary condition $b_-^\prime \in \mathcal{I}_{j_-^\prime}$.  Then there exists  $\mathcal{A}<\infty$ such that \begin{equation}
        \left\lVert \mu^{\lbrace p, q\rbrace}_{\Lambda,\beta, b_-, b_+} -\mu^{\lbrace p, q\rbrace}_{\Lambda^\prime ,\beta, b^\prime_-, b_+} \right\rVert_{\mathrm{TV}} < \mathcal{A} \rho^{(p-j_-)/L}. \label{eq:TV_right_moving_contraction}
    \end{equation}
\end{cor}
 \begin{proof}
    Consider the distributions $|\nu\rangle := |b_- - \mathcal{R}(\xi_{j_-})\rangle$ and \begin{equation}
            |\nu^\prime \rangle := P_{j_--1} \cdots P_{j_-^\prime} |b_-^\prime  - \mathcal{R}(\xi_{j_-^\prime} )\rangle.
        \end{equation}  In part using Proposition \ref{prop:bounded_V}, notice $\langle V|\nu\rangle = 1$ and \begin{equation}
            \langle V|\nu^\prime\rangle \le \frac{J}{1-r^{-1}}.
        \end{equation} are both bounded. Moreover,  \begin{equation}
            W(|\nu\rangle, |\nu^\prime\rangle) \le 1 + \gamma + \frac{\gamma J}{1-r^{-1}} =: \mathcal{A} \rho^{1-L}. \label{eq:mathcal_A_def}
        \end{equation} Then apply Lemma \ref{lem:W_contraction} to obtain the exponential shrinking of $W$ with further evolution.  Applying Proposition \ref{prop:TV} shows that the ``time" evolution of $|\nu\rangle$ and $|\nu^\prime\rangle$ converge exponentially quickly in total variation.\footnote{We added the unimportant factor $\rho^{1-L}$ to \eqref{eq:mathcal_A_def} because Lemma \ref{lem:W_contraction} requires steps of size $L$, while we do not impose that restriction in \eqref{eq:TV_right_moving_contraction}.}  

        It now remains to ``lift" this convergence theorem to the total variation convergence of the distributions on the on-site state space $\mathbb{Z}^2$, rather than $\widetilde{\mathcal{I}}\times \mathbb{Z}_{\ge 0}$.  Here we simply reverse the construction of Definition \ref{def:auxiliary} as follows: to convert the Markov chain $(\ldots, X_i, \ldots)$ to $(\ldots, x_i, \ldots)$, take each $X_i = (\widetilde{d}_i,0) \rightarrow \widetilde{d}_i+\mathcal{R}(\xi_i)$.  For each sequence $(\ldots, X_{i},\ldots, X_{i+n-1}, \ldots)$ with $X_{i+j} = (\widetilde{d}_{i+n}, n-j)$, sample $(x_i, \ldots, x_{i+n-1})$ from the distribution \begin{align}
            \mu^{\text{repair}}_\beta &(x_i,\ldots, x_{i+n-1} | x_{i+n},x_{i-1}) \notag \\
            &= \mathbb{I}(x_{i+\ell}\in \mathcal{N}_{i+\ell} \text{ for } 0\le \ell<n)  \frac{\langle x_{i+n}|K_\beta |x_{i+n-1}\rangle \cdots \langle x_i |K_\beta |x_{i-1}\rangle }{\langle x_{i+n}|F_{i+n,i-1}|x_{i-1}\rangle }
        \end{align}
        conditioned on the endpoints $X_{i-1},X_{i+n}\in \widetilde{\mathcal{I}}$, where we again used \eqref{eq:rounded_displaced_coordinates}.   For the path coupling problem, the sampling from $\mu^{\text{repair}}$ is independent for each chain until the flag $f=*$ indicates that they coupled \emph{before} the excursion to the non-ideal region.  This then ensures that the boundary conditions of $\mu_\beta^{\text{repair}}$ match for both chains, safely allowing us to use the same draw for sampling.   This establishes \eqref{eq:TV_right_moving_contraction} for the desired  $\mu^{\lbrace p, q\rbrace}_{\Lambda,\beta, b_-, b_+} $.
    \end{proof}

\subsection{Proof of Theorem \ref{thm:exist}}\label{sec:thmproof}
Corollary \ref{cor:left_boundary} is very suggestive that the Gibbs state is independent of details (other than $c$) associated with the left boundary condition.  We'll make that precise shortly.  First, we observe the following corollary, which follows immediately from the arguments which led to Corollary \ref{cor:left_boundary}, but now applied to an analogous inhomogeneous Markov chain defined by reading the Gibbs state from right to left: 
\begin{cor}\label{cor:right_boundary}
        For any $j_-<p<q<j_+$, let $\mu^{\lbrace p, q\rbrace}_{\Lambda,\beta, b_-,b_+}$ denote the Gibbs state restricted to the interval $\lbrace p, \ldots, q\rbrace$, with $b_\pm\in \mathcal{I}_{j_\pm}$.    Consider changing the domain to $\Lambda^\prime = \lbrace j_- ,\ldots ,j_+^\prime\rbrace$ with $j_+^\prime > j_+$ and new boundary condition $b_+^\prime \in \mathcal{I}_{j_+^\prime}$. Then there exists a constant $\mathcal{A}^\prime<\infty$ and $0<\rho^\prime<1$ such that \begin{equation}
        \lVert \mu^{\lbrace p, q\rbrace}_{\Lambda,\beta, b_-, b_+} -\mu^{\lbrace p, q\rbrace}_{\Lambda^\prime ,\beta, b_-, b_+^\prime} \rVert_{\mathrm{TV}} < \mathcal{A}^\prime  (\rho^\prime)^{(j_+-q)/L}. \label{eq:TV_left_moving_contraction}
    \end{equation}
\end{cor}
\noindent The constants $\mathcal{A}^\prime$ and $\rho^\prime$ may differ as we effectively use the matrix $A^{-1}$ instead of $A$ when reading the chain backwards, and the states $x_n$ are effectively replaced by $Ax_n$ when reading the chain backwards.  In any case, take $\mathcal{A}_* := \max(\mathcal{A},\mathcal{A}^\prime)$ and $\rho_* := \max(\rho,\rho^\prime)$.   Critically we can now combine Corollary \ref{cor:left_boundary} with Corollary \ref{cor:right_boundary} to obtain: \begin{cor} \label{cor:DLR_exists}
    Take any nested sequence of growing regions $\lbrace j_-, \ldots, j_+\rbrace$ with corresponding boundary conditions $b_\pm \in \mathcal{I}_{j_\pm}$. The measure \begin{equation}
        \lim_{j_- \rightarrow -\infty} \lim_{j_+\rightarrow \infty} \mu^{\lbrace p,q\rbrace}_{\lbrace j_-,\ldots, j_+\rbrace, \beta, b_-,b_+} =: \mu^{c, \lbrace p, q\rbrace}_\beta  \label{eq:mu_c_beta_exist}
    \end{equation}exists.  The order of limits does not matter.
\end{cor}
\begin{proof}
    Taking $(j_-,j_+)$ to $(j_-^\prime <j_-, j_+^\prime > j_+)$, we see that \begin{equation} \label{eq:enlarging-window}
        \lVert \mu^{\lbrace p,q\rbrace}_{\lbrace j_-,\ldots, j_+\rbrace, \beta, b_-,b_+} - \mu^{\lbrace p,q\rbrace}_{\lbrace j_-^\prime,\ldots, j_+^\prime\rbrace, \beta, b_-^\prime,b_+^\prime} \rVert_{\mathrm{TV}} < \mathcal{A}_* \left[ \rho_*^{(p-j_-)/L} + \rho_*^{(j_+-q)/L} \right].
    \end{equation}
    As we take all four $|j_\pm|$ and $|j_\pm^\prime|$ to $\infty$, the right hand side vanishes; hence the sequence of measures $\mu^{\lbrace p,q\rbrace}_{\lbrace j_-,\ldots,j_+\rbrace, \beta, j_-, j_+}$ is Cauchy in total variation.  Since probability measures on the countable state space $(\mathbb{Z}^2)^{q-p+1}$ are complete in total variation, the limit \eqref{eq:mu_c_beta_exist} must exist.
\end{proof}

\eqref{eq:thm_exist_main} is implied by \eqref{eq:mu_c_beta_exist}.  Moreover, Corollary \ref{cor:DLR_exists} also shows the independence of $\mu^{c, \lbrace p, q\rbrace}_\beta $ on the rounding map, as claimed in Theorem \ref{thm:exist}.  The reason is simple: adjusting the rounding map at one end simply corresponds to differing choices of $b_\pm$ and we have established that memory of the detailed choice of $b_\pm$ is lost exponentially as $|j_\pm| \rightarrow \infty$.  Applying the same argument separately to each side verifies independence of $\mu^{c, \lbrace p, q\rbrace}_\beta$ on the rounding map.   Moreover, since marginalization causes only contraction in total variation, for any $p<p^\prime<q^\prime < q$, marginalizing $\mu^{c, \lbrace p, q\rbrace}_\beta $ to the window $\lbrace p^\prime ,\ldots, q^\prime\rbrace$ approaches the same limit as $\mu^{c, \lbrace p^\prime, q^\prime\rbrace}_\beta$ obtained directly from Corollary \ref{cor:DLR_exists}.  By the Kolmogorov extension theorem, we deduce the existence of the limiting measure $\mu^c_\beta$ on the entire state space $\Omega$.

Next, let us check that $\mu^c_\beta$ is equivariant under the symmetries.  This follows immediately from the limiting process above.  Firstly, the entire finite-size measures obey \begin{equation}
    T_* \mu^{\lbrace p,q\rbrace}_{\lbrace j_-,\ldots, j_+\rbrace, \beta, b_-,b_+} = \mu^{\lbrace p-1,q-1\rbrace}_{\lbrace j_--1,\ldots, j_+-1\rbrace, \beta, b_-,b_+}.
\end{equation}Since $b_\pm$ were unchanged, but all coordinates shifted, this changes the ideal solution: \begin{equation}
    \xi_j = A^{-j }c \rightarrow \xi_{j+1} = A^{-(j+1)}c = A^{-j}(A^{-1}c).
\end{equation}
The rounding map $\mathcal{R}$ keeps these boundary conditions in the same ideal region as if we had shifted $c \rightarrow A^{-1}c$.  Since the limiting distribution is independent of $\mathcal{R}$, we deduce that $T_*\mu^c_\beta = \mu^{A^{-1}c}_\beta$.   Analogously, 
\begin{equation}
    \left((U_a)_* \mu^{\lbrace p,q\rbrace}_{\lbrace j_-,\ldots, j_+\rbrace, \beta, b_-,b_+}\right) = \mu^{\lbrace p,q\rbrace}_{\lbrace j_-,\ldots, j_+\rbrace, \beta, b_- + A^{-j_-}a,b_+ + A^{-j_+}a},
\end{equation}
which is identical, for any $\mathcal{R}$, to shifting the ideal solution \begin{equation}
    \xi_j = A^{-j}c \rightarrow \xi_j = A^{-j}(c+a).
\end{equation}
Crucially, we see that there is a natural group action of $\mathbb{Z}^2\rtimes \mathbb{Z}$ on the $c$ label of the DLR state, and by our construction of the DLR state via the Cauchy sequence, the helical group symmetry acts only on the label $c$ in the limit.  This establishes \eqref{eq:equivariant}, the desired equivariance properties.

Lastly, it is obvious that the DLR equation \eqref{eq:DLR} is obeyed, since \begin{equation}
     \mu^{\lbrace p-1,q+1\rbrace}_{\lbrace j_-,\ldots, j_+\rbrace, \beta , b_-, b_+} (x_p,\ldots, x_q|x_{p-1},x_{q+1}) = \mu_{\lbrace p-1,p,\ldots, q+1\rbrace,\beta,x_{p-1},x_{q+1}}(x_p,\ldots, x_q)
\end{equation}
holds by construction for each $j_-<p\le q<j_+$; hence it also holds in the infinite-volume limit. 

\subsection{Proof of Corollary \ref{cor:continuity}}\label{sec:corproof-2}
The technology used in the proof of Theorem \ref{thm:exist} allows us to prove many of the remaining open corollaries from Section \ref{sec:mucbeta}.  We start with Corollary \ref{cor:continuity}.

Fix any $\epsilon>0$; \eqref{eq:enlarging-window} implies that there exists a large enough $N$ such that for any rounding map $\mathcal{R}$,  \begin{equation}
         \left\lVert \mu^c_\beta (x_{j\in \Lambda_0} ) - \mu^{c,\mathcal{R}}_{\beta, N}(x_{j\in \Lambda_0}) \right\rVert_{\mathrm{TV}}  < \epsilon.
     \end{equation}Here $\mu^{c,\mathcal{R}}_{\beta, N} = \mu_{\Lambda_N, \beta, \mathcal{R}(A^Nc), \mathcal{R}(A^{-N}c)}$ (recall Definition \ref{def:helical-sequence}).  For $c^\prime$ sufficiently close to $c$ at any fixed $N$,  there exists a rounding map such that $\mathcal{R}(A^{\pm N}c) = \mathcal{R}(A^{\pm N}c^\prime)$.\footnote{A weaker condition here is also sufficient: we only need that there exist endpoints $b_\pm$ that lie in the ideal regions $\mathcal{I}_{\pm N}$ defined in Definition \ref{def:ideal} for both $c$ and $c^\prime$.}  As such, $\mu^{c,\mathcal{R}}_{\beta,N} = \mu^{c^\prime,\mathcal{R}}_{\beta,N} $ are identical measures.  Now using \eqref{eq:thm_exist_main} twice and applying the triangle inequality, we know that \begin{align}
         \left\lVert \mu^c_\beta(x_{j\in \Lambda_0}) - \mu^{c^\prime}_\beta(x_{j\in \Lambda_0})\right\rVert_{\mathrm{TV}} &\le \left\lVert \mu^c_\beta(x_{j\in \Lambda_0}) - \mu^{c,\mathcal{R}}_{\beta,N}(x_{j\in \Lambda_0})\right\rVert_{\mathrm{TV}} + \left\lVert \mu^{c^\prime}_\beta(x_{j\in \Lambda_0}) - \mu^{c^\prime,\mathcal{R}}_{\beta,N}(x_{j\in \Lambda_0})\right\rVert_{\mathrm{TV}} \notag \\
         &< 2\epsilon. \label{eq:TV-c-cprime}
     \end{align}
     Thus at any fixed $\epsilon$, all $c^\prime$ sufficiently close to $c$ obey \eqref{eq:TV-c-cprime}.  Continuity follows from the limit $\epsilon \rightarrow 0$.

\subsection{Proof of Corollary \ref{cor:deviation}}\label{sec:corproof}
The idea is, once again, that Theorem \ref{thm:sound} exponentially suppresses large excursions of the ideal path.  The corollary is trivial if $x_j\in\mathcal{I}_j$; we only need to track what happens when $x_j\in \mathcal{N}_j$.     More concretely, it is useful to temporarily work with the finite-volume measure $\mu^c_{\lbrace j_-,\ldots, j_+\rbrace, \beta, b_-, b_+}$, using Theorem \ref{thm:exist} to take the limits $j_\pm \rightarrow \pm \infty$ at the end.  

Suppose we pin the trajectory during its excursion out of the ideal region to hit $x_j\in \mathcal{N}_j$ conditioned on the existence of $p$ and $q$ obeying $j_-\le p<j<q\le j_+$.  Assuming for the moment that we know the values $x_p = c_-$ and $x_q = c_+$, we can easily see that the sum over all anchored paths is \begin{equation}
    \langle c_+| F_{q,p}^{(j)}(x) | c_-\rangle := \sum_{x_i \in \mathcal{N}_i, \; p<i<q, \; i\ne j} \prod_{i=p}^{q-1}\langle x_{i+1}|K_\beta |x_i\rangle \le \mathrm{e}^{\beta C} \sigma^{q-p-2} \mathrm{e}^{-\beta \eta  \lVert x_j - \xi_j\rVert_1}
\end{equation}
as a straightforward generalization of \eqref{eq:max-of-F}.  Analogously to \eqref{eq:P-excursion}, we have \begin{equation}
    \frac{\langle c_+| F_{q,p}^{(j)}(x) | c_-\rangle}{\langle c_+|G_{q,p}|c_-\rangle }\le \frac{\mathrm{e}^{\beta C}}{Q\zeta }\theta^{q-p-2} \mathrm{e}^{-\beta \eta  \lVert x_j - \xi_j\rVert_1}.
\end{equation}
We then have, analogous to \eqref{eq:P-excursion}, \begin{align}
    \mu^c_{\lbrace j_-,\ldots, j_+\rbrace, \beta, b_-, b_+}(x_j) &\le \sum_{p=j_-}^{j-1}\sum_{q=j+1}^{j_+}\sum_{c_-\in \mathcal{I}_p }\sum_{c_+\in \mathcal{I}_q } \frac{\langle b_+|K_\beta^{j_+-q}|c_+\rangle \langle c_+|F^{(j)}_{q,p}(x)|c_-\rangle \langle c_-|K_\beta^{p-j_-}|b_-\rangle }{\langle b_+|K_\beta^{j_+-j_-}|b_-\rangle } \notag \\
    &\le \sum_{p=j_-}^{j-1}\sum_{q=j+1}^{j_+} \max_{c_+\in \mathcal{I}_q, c_-\in \mathcal{I}_p}\frac{\langle c_+| F_{q,p}^{(j)}(x) | c_-\rangle}{\langle c_+|G_{q,p}|c_-\rangle } < \frac{\mathrm{e}^{\beta C}}{Q\zeta(1-\theta)^2} \mathrm{e}^{-\beta \eta  \lVert x_j-\xi_j\rVert_1}.
\end{align}
Using Theorem \ref{thm:exist} to take $j_\pm \rightarrow \pm \infty$, we deduce \eqref{eq:exponential-decay} with $\alpha = \beta \eta$; converting between $\xi_j$ and $\mathcal{R}(A^{-j}c)$ only modifies the constant prefactor by the triangle inequality.

\subsection{Proof of Theorem \ref{thm:extremal}}\label{sec:extremal-proof}
Lastly we can prove Theorem \ref{thm:extremal}.  For any $\Lambda \in \mathcal{S}$, let $f$ be a bounded $\mathcal{F}_\Lambda$-measurable function, i.e. $\lVert f\rVert_\infty < \infty$.   Recall that $M$ represented the tolerance in defining the ideal regions $\mathcal{I}_j$ in Definition \ref{def:ideal}.  Pick a helical sequence (Definition \ref{def:helical-sequence}) with large enough $N$ so that $\Lambda \subset \Lambda_N$ for all $N$ in the sequence.    Theorem \ref{thm:exist} shows that \begin{equation}
        \sup_{x \in G_{N,M}} \left| \mu_{\Lambda_N,\beta ,x_{-N},x_{+N}}(f) - \mu^c_\beta(f)\right| \le \epsilon_{N,M}
    \end{equation}
    with \begin{equation}
        \lim_{N\rightarrow \infty } |\epsilon_{N,M}| =  0.
    \end{equation}

     Define the event \begin{equation}
        G_{N,M} := \lbrace x_{\pm N} \in \mathcal{I}_{\pm N}\rbrace 
    \end{equation}and define the $\sigma$-algebras $\mathcal{F}_N := \sigma(\lbrace x_j : |j| \ge N\rbrace)$.  We know that \begin{align}
       \frac{1}{\lVert f\rVert_\infty} \left\lVert \mu^c_\beta (f | \mathcal{F}_N) - \mu^c_\beta (f) \right\rVert_{\mathrm{L}^1(\mu^c_\beta)} &\le  \mu^c_\beta({G_{N,M}}^{\mathrm{c}}) \notag \\
       &\;\;\;\;\; + \sup_{x \in G_{N,M}} \left| \mu_{\Lambda_N,\beta ,x_{-N},x_{+N}}(f) - \mu^c_\beta(f)\right| \mu^c_\beta(G_{N,M}) \notag \\
        & \le   \epsilon_{N,M} + 2\delta_M
    \end{align}
    where $\delta_M$ is an exponentially decaying constant in $M$ due to Corollary \ref{cor:deviation}.  Therefore, \begin{equation}
       \lim_{M\rightarrow \infty }  \lim_{N\rightarrow \infty} \left\lVert \mu^c_\beta (f | \mathcal{F}_N) - \mu^c_\beta (f) \right\rVert_{\mathrm{L}^1(\mu^c_\beta)} = 0.
    \end{equation}
    This establishes the $\mathrm{L}^1$-convergence of $\mu^c_\beta (f | \mathcal{F}_N) $ to  $\mu^c_\beta(f)$. The backwards martingale convergence theorem then implies that $\mu^c_\beta$-almost surely, \eqref{eq:extremal-property} holds.  By Corollary \ref{cor:extremal-property}, $\mu^c_\beta$ is an extremal DLR state.

\section{Regularizing the model on a finite-dimensional space}\label{sec:finite-size}
In this section, we will develop a regularization of the model in Definition \ref{def:helicalcode} based on a finite-dimensional global state space.  Clearly, then, there is a unique Gibbs measure.  One can still identify the imprint of spontaneous symmetry breaking by looking for measure concentration.  This is useful from a physics perspective, because one can also analyze more familiar models like the 2d Ising model in this way; therefore, we will arrive at an independent justification for our claim that the symmetry breaking in the infinite chain should be understood as spontaneous.  
\subsection{A natural lattice regularization}
We first need to choose a clever finite regularization of the model, and then adapt Theorem \ref{thm:exist} to analyze the unique Gibbs measure.
\begin{defn}
    For any $L\in \mathbb{Z}_+$, and hyperbolic symmetric $A$, define \begin{subequations}
        \begin{align}
            M_L &:= \left(A^L - I\right) \mathbb{Z}^2, \\
            G_L &:= \mathbb{Z}^2/M_L.
        \end{align}
    \end{subequations}
\end{defn}
\noindent  It is also useful to observe that since $A\in \mathrm{SL}(2,\mathbb{Z})$, $A M_L = M_L$.   So $A$ induces an automorphism on the coset space $G_L$, which is easily seen to be of order $L$.
\begin{prop}
    There exists a sequence of integers $r_1,r_2,\ldots$ obeying \begin{equation}
        r_L = \exp(\mathrm{\Theta}(L)) \label{eq:r_L}
    \end{equation}
    such that $M_L \subseteq r_L\mathbb{Z}^2$.
\end{prop}
\begin{proof}
    Let $t := \mathrm{tr}(A)>2$.  The Cayley-Hamilton Theorem provides the useful recursion relation \begin{equation}
        A^2 = tA - I. \label{eq:cayley-hamiltonian}
    \end{equation}
    Now consider the arithmetic sequence defined by $p_0=0$, $p_1=1$ and for $n>1$, \begin{equation}
        p_{n+1}= tp_n - p_{n-1}. \label{eq:p-arithmetic}
    \end{equation}
    Combining \eqref{eq:cayley-hamiltonian} and \eqref{eq:p-arithmetic} we obtain \begin{equation}
        A^n = p_n A - p_{n-1}I. \label{eq:An-recursion}
    \end{equation}
    $A^{-1}$ also obeys \eqref{eq:cayley-hamiltonian} and \eqref{eq:An-recursion};  using the further identity $A^{-1} = tI-A$ (which is readily checked by explicitly inverting $2\times 2$ matrices!), we deduce \begin{equation}
        A^{-n} = p_{n+1}I - p_n A.\label{eq:An-inv-recursion}
    \end{equation}
    Combining \eqref{eq:An-recursion} and \eqref{eq:An-inv-recursion} we find the elegant identity: \begin{equation}
        (A^{2m}-I)\mathbb{Z}^2 = \left(A^m - A^{-m} \right)A^m\mathbb{Z}^2 = p_m\left(A-A^{-1}\right)\mathbb{Z}^2 \subset p_m \mathbb{Z}^2 \label{eq:ML-even}
    \end{equation}
    for integer $m$; above we used that $A$ is unimodular.  Similarly, \begin{equation}
        (A^{2m+1}-I)\mathbb{Z}^2 = \left(A^{m+1} - A^{-m} \right)A^m\mathbb{Z}^2 = (p_m+p_{m+1})\left(A-I\right)\mathbb{Z}^2 \subset (p_m+p_{m+1}) \mathbb{Z}^2. \label{eq:ML-odd}
    \end{equation}
    \eqref{eq:ML-even} and \eqref{eq:ML-odd} establish part of the claim: $r_{2m} = p_m$ and $r_{2m+1} = p_m + p_{m+1}$.

    It remains to show that $p_m = \exp(\mathrm{\Theta}(m))$. Notice that $p_{n+1}>p_n$, which follows from $t>2$ and can be proved inductively from \eqref{eq:p-arithmetic} given the base case $p_2=t>1=p_1$.  Thus for $n\ge 2$, \begin{equation}
        t^{n-1}\ge p_n>(t-1)^{n-1},
    \end{equation}which completes the proof.
\end{proof}
\eqref{eq:ML-even} and \eqref{eq:ML-odd} reveal that the lattices $M_L$ are not too distorted:  they are simply rescaled ``simple" lattices $(A-A^{-1})\mathbb{Z}^2$ or $(A-I)\mathbb{Z}^2$ by exponentially growing prefactors $r_L$, alternating between two lattice shapes depending on whether $L$ is even or odd.  This has a useful implication:
\begin{prop}\label{prop:size-GL}
    For some O(1) constant $J$ depending only on $A$, \begin{equation}
        r_L^2 \le |G_L| \le J r_L^2. \label{eq:JrL2}
    \end{equation}
\end{prop}
\begin{proof}
    Suppose we have the following Smith decompositions:  for $U_{1,2}$ and $V_{1,2}$ unimodular: \begin{subequations}\begin{align}
        U_1 (A-I)V_1  &= \left(\begin{array}{cc} d_{1,1} &\ 0 \\ 0 &\ d_{1,2} \end{array}\right), \\
        U_2 (A-A^{-1})V_2  &= \left(\begin{array}{cc} d_{2,1} &\ 0 \\ 0 &\ d_{2,2} \end{array}\right).
    \end{align}\end{subequations}
    Define $J:=\max(d_{1,1}d_{1,2},d_{2,1}d_{2,2})$.  \eqref{eq:JrL2} follows from the identity \begin{equation}
        |G_L| = \left| \det(A^L-I)\right| \propto r_L^2,
    \end{equation} with the proportionality factor given by the appropriate factor of Smith factors above, along with the fact that all Smith factors $d_{1/2,1/2}\ge 1$.
\end{proof}

\subsection{Finite-size model}
\begin{defn}[Lee 1-norm] \label{def:lee}
    Given $x\in G_L$ and its arbitrary\footnote{The key point is that $\lVert x-y\rVert_{\mathrm{L}}$ is a well-defined notion of distance on the coset $G_L$: it is easy to see that the value is independent of the representatives $x,y \in \mathbb{Z}^2/M_L$.} lift $\widetilde{x} \in \mathbb{Z}^2$, define the \textbf{Lee 1-norm} \begin{equation}
        \lVert x \rVert_{\mathrm{L}} := \min_{y\in M_L} \lVert \widetilde{x}-y\rVert_1.
    \end{equation}
\end{defn}
\begin{prop}\label{prop:GL-middle}
    Let \begin{equation}\label{eq:define_gamma}
        \gamma := \frac{1}{2(1+\lVert A\rVert_{1\rightarrow 1})}.
    \end{equation}
    If a lift $\widetilde{x}$ satisfies $\lVert \widetilde{x}\rVert_1 <\gamma r_L$, then it is the (unique) smallest lift in its coset $\widetilde{x}+M_L$.  As such,  $\lVert \widetilde{x}\rVert_1 = \lVert x\rVert_{\mathrm{L}}$.  Similarly,  $\lVert A\widetilde{x}\rVert_1 = \lVert Ax\rVert_{\mathrm{L}}$.
\end{prop}
\begin{proof}
       This follows immediately from Definition \ref{def:lee}: since $M_L\subseteq r_L\mathbb{Z}^2$, any non-zero $z\in M_L$ has $\lVert z\rVert_1 \ge r_L$.
\end{proof}
\begin{defn}[Finite-length helical code]
    Consider the state space $\Omega = (G_L)^L$ labeled by configurations $(x_0,\ldots, x_{L-1})$, with $x_L=x_0$ denoting  periodic boundary conditions. Define the finite-length helical code Hamiltonian \begin{equation}
        H(x_0,\ldots, x_{L-1}) := \sum_{j=0}^{L-1}\lVert x_j - Ax_{j+1}\rVert_{\mathrm{L}}.
    \end{equation} 
    This Hamiltonian is invariant under \begin{subequations}
        \begin{align}
            T \cdot (x_0,\ldots, x_{L-1}) &= (x_1,\ldots, x_{L-1},x_0), \\
            U_a \cdot (x_0,\ldots, x_{L-1}) &= (x_0+a, x_1+A^{-1}a, \ldots, x_{L-1}+A^{-(L-1)}a), \label{eq:finite-size-helical}
        \end{align}
    \end{subequations}
    which generate the discrete helical symmetry group $G_L\rtimes \mathbb{Z}_L$, with commutation relation \eqref{eq:helicalcommutation}.
\end{defn}
\noindent We remind the reader that in equations like \eqref{eq:finite-size-helical}, $A^{-i}a$ are interpeted as well-defined elements of the coset space $G_L$.

\subsection{Concentration of the unique Gibbs measure}
We now establish measure concentration of the Gibbs state on this large, but finite-dimensional, state space.  This is a standard way to detect spontaneous symmetry breaking with finite system sizes, and is arguably much simpler to follow than the DLR construction above!  The discussion below is a relatively direct generalization of results for a simpler model that are presented in \cite{leecode1}.

\begin{thm}[Spontaneous symmetry breaking]\label{thm:finite}
    For any $0<\beta<\infty$, there exists a $\beta$-dependent constant $L_0$ such that for all $L\ge L_0$, the following holds.   Consider the unique Gibbs measure \begin{equation}
        \mu_\beta(x) := \frac{1}{Z_\beta}\mathrm{e}^{-\beta H(x)},
    \end{equation}
    with standard normalization factor $Z_\beta$.   There exists a decomposition of the state space \begin{equation}
        \Omega = \Omega_{\text{bn}} \sqcup \bigsqcup_{a\in G_L}\Omega_a,
    \end{equation}
    such that \begin{equation}
        U_b\cdot \Omega_a = \Omega_{a+b} \label{eq:Omega_a_group}
    \end{equation} and, for some $L$-independent/$\beta$-dependent constants $0<C,c<\infty$,  \begin{equation}
        \mu_\beta (\Omega_a) = \frac{1}{|G_L|} \left(1 - C \mathrm{e}^{-\beta c r_L}\right). \label{eq:mu_Omega_a}
    \end{equation}Moreover, the distance between any two of these ``symmetry-broken" clusters is large: for $a\ne a^\prime$, there exists $L$-independent constant $\gamma>0$ such that \begin{equation}
        \min_{x\in \Omega_a, y\in \Omega_{a^\prime}} \sum_{i=0}^{L-1}\lVert x_i-y_i\rVert_{\mathrm{L}} \ge \frac{\gamma r_L}{3}. \label{eq:cluster-distance-big}
    \end{equation}
\end{thm}
\begin{proof}
Define the following two transfer matrices: for $x,y\in G_L$,
            \begin{subequations}\begin{align}
                \langle y|\mathcal{K}_\beta|x\rangle &:= \mathrm{e}^{-\beta \lVert x-Ay\rVert_{\mathrm{L}}}, 
                 \\
                \langle y|\widehat{\mathcal{K}}_\beta|x\rangle &:= \sum_{z\in M_L} \mathrm{e}^{-\beta \lVert \widetilde{x}-A\widetilde{y} + z\rVert_{1}}.
            \end{align}\end{subequations}
In the second equation, we used any valid lifts $\widetilde{x},\widetilde{y}\in \mathbb{Z}^2$. 
Notice that $\langle x|\widehat{\mathcal{K}}_\beta|y\rangle > \langle x|\mathcal{K}_\beta|y\rangle$.  On the ring with periodic boundary conditions, we have partition function \begin{equation}
    Z_\beta := \mathrm{tr}\left(\mathcal{K}_\beta^L\right).
\end{equation}

The next two lemmas show us that we may approximately replace $\mathcal{K}_\beta$ with $\widehat{\mathcal{K}}_\beta$ when calculating $Z_\beta$, and that the resulting calculation has a very elegant interpretation as a problem on the original state space $\mathbb{Z}^2$.
\begin{lem} \label{lem:Kbeta}
        There exists $L$-independent constant $0<C_0<\infty$ such that: \begin{equation}
        0< \frac{\mathrm{tr}\left(\widehat{\mathcal{K}}_\beta^L\right)}{\mathrm{tr}\left(\mathcal{K}_\beta^L\right)} -  1  < C_0 \mathrm{e}^{-\beta r_L/4}. \label{eq:good-approx}
        \end{equation}
    \end{lem}
    \begin{proof}
        Notice that if \begin{equation}
            |1_{G_L}\rangle := \sum_{x\in G_L} |x\rangle,
        \end{equation}
        then analogous to \eqref{eq:kappa_beta}, we have $\widehat{\mathcal{K}}_\beta |1_{G_L}\rangle = \kappa_\beta |1_{G_L}\rangle $.  Next, we observe that because $M_L \subseteq r_L \mathbb{Z}^2$, for any two distinct points $z_{1,2}\in M_L$, \begin{equation}
            \max\left(\lVert x- z_1\rVert_1, \lVert x- z_2\rVert_1\right) \ge \frac{r_L}{2}\label{eq:rL_big_distance}
        \end{equation}since $\lVert z_1-z_2\rVert_1 \ge r_L$.  Combining \eqref{eq:rL_big_distance} with \eqref{eq:sigma_inequality}, we obtain  \begin{equation}
            \langle x| \left( \widehat{\mathcal{K}}_\beta -  \mathcal{K}_\beta \right)|y\rangle < \frac{4r_L \mathrm{e}^{-\beta (r_L-2)/2}}{(1-\mathrm{e}^{-\beta})^2} < \epsilon_* \mathrm{e}^{-\beta r_L/3}
        \end{equation}
        for some fixed $L$-independent $\epsilon_*$.  Then using the positivity of every matrix element of $\widehat{\mathcal{K}}_\beta$ and $\mathcal{K}_\beta$, we obtain \begin{align}
            0<\mathrm{tr}\left(\widehat{\mathcal{K}}_\beta^L\right) - \mathrm{tr}\left(\mathcal{K}_\beta^L\right) &< L \mathrm{tr} \left( \left(\widehat{\mathcal{K}}_\beta - \mathcal{K}_\beta\right)\widehat{\mathcal{K}}_\beta^{L-1} \right) < L \epsilon_* \mathrm{e}^{-\beta r_L/3} \langle 1_{G_L}| \widehat{\mathcal{K}}_\beta^{L-1}|1_{G_L}\rangle  \notag \\
            &< L |G_L| \epsilon_* \mathrm{e}^{-\beta r_L/3} \kappa_\beta^{L-1 } = \exp\left(\mathrm{\Theta}(L) - \frac{\beta r_L}{3}\right). \label{eq:good-approx-2}
        \end{align}
        Lastly, observe \begin{equation}
         \mathrm{tr}\left(\mathcal{K}_\beta^L\right) > 1 \label{eq:KLbeta1}
    \end{equation}
    because we may consider the sole configuration $x_i=0$ which contributes 1 to the partition function. 
        \eqref{eq:good-approx} follows from combining \eqref{eq:good-approx-2} and \eqref{eq:KLbeta1} with \eqref{eq:r_L},  along with Proposition \ref{prop:size-GL} and elementary bounds on products of exponential functions.
    \end{proof}

    \begin{lem}\label{lem:lift-to-Z2}
        Given transfer matrix $K_\beta$ (Definition \ref{def:transfer-matrix}) for the on-site state space $\mathbb{Z}^2$, \begin{equation}
            \mathrm{tr}\left(\widehat{\mathcal{K}}_\beta^L\right) = |G_L| \cdot  \mathrm{tr}\left(K_\beta^L\right). \label{eq:lift-to-Z2}
        \end{equation}
    \end{lem}
    \begin{proof}
        Write \begin{align}
            \mathrm{tr}\left(\widehat{\mathcal{K}}_\beta^L\right) &= \sum_{y_0,\ldots,y_{L-1}\in G_L}\sum_{z_0,\ldots,z_{L-1}\in M_L} \prod_{i=0}^{L-1} \mathrm{e}^{-\beta \lVert y_i - Ay_{i+1}+z_i\rVert_1} \notag \\
            &= \sum_{y_0,\ldots,y_{L-1}\in G_L}\sum_{z_0,\ldots,z_{L-1}\in M_L} \mathrm{e}^{-\beta \lVert y_0(I-A^L)+(z_0+Az_1+\cdots + A^{L-1}z_{L-1}) - (Ad_1 +\cdots + A^{L-1}d_{L-1}) \rVert_1} \prod_{i=1}^{L-1} \mathrm{e}^{-\beta \lVert d_i\rVert_1} \notag \\
            &= |G_L| \sum_{z_0^\prime \in M_L}\sum_{d_1,\ldots, d_{L-1}\in \mathbb{Z}^2} \mathrm{e}^{-\beta \lVert z_0^\prime  - (Ad_1 +\cdots + A^{L-1}d_{L-1}) \rVert_1} \prod_{i=1}^{L-1} \mathrm{e}^{-\beta \lVert d_i\rVert_1} \label{eq:compare1}
        \end{align}
        with $y_0=y_L$ and $z_0=z_L$ enforcing periodic boundary conditions.  In the second line, we defined $d_i = y_i + z_i - Ay_{i+1}$, which can take arbitrary values in $\mathbb{Z}^2$.   In the third line, we notice that $y_0 (I-A^L)\in M_L$, so we absorb all such factors into a sum over $z_0^\prime \in M_L$ by shifting $z_0$, leaving a free sum over $y_0$ giving overall factor $|G_L|$.   Now,    \begin{align}
            \mathrm{tr}\left(K_\beta^L\right) &= \sum_{x_0 \in \mathbb{Z}^2}\sum_{d_1,\ldots, d_{L-1}\in \mathbb{Z}^2} \mathrm{e}^{-\beta \lVert x_0 - (Ad_1 + \cdots +A^{L-1}d_{L-1}+A^Lx_0)\rVert_1}\prod_{i=1}^{L-1}\mathrm{e}^{-\beta \lVert d_i\rVert_1} \notag \\
            &= \sum_{z_0^{\prime\prime} \in  M_L}\sum_{d_1,\ldots, d_{L-1}\in \mathbb{Z}^2} \mathrm{e}^{-\beta \lVert z_0^{\prime\prime} - (Ad_1 + \cdots +A^{L-1}d_{L-1})\rVert_1}\prod_{i=1}^{L-1}\mathrm{e}^{-\beta \lVert d_i\rVert_1} \label{eq:compare2}
        \end{align}
        using an analogous change of variables as above: $d_i=x_i - Ax_{i+1} \in \mathbb{Z}^2$.  Comparing \eqref{eq:compare1} and \eqref{eq:compare2} we obtain \eqref{eq:lift-to-Z2}.
    \end{proof}

    Lemma \ref{lem:lift-to-Z2} is very suggestive:  intuitively,  $\mathrm{tr}(K_\beta^L)$ should look like the $c=0$ DLR state from Theorem \ref{thm:exist}, because that is the only state that avoids exponentially growing $x_j$ near the boundaries.  And indeed, using a straightforward generalization of the proof of Theorem \ref{thm:sound}, we find that \begin{equation}
        H_{\mathbb{Z}^2, \text{pbc}}(x_0,\ldots, x_{L-1}) \ge \eta \sum_{i=0}^{L-1}\lVert x_i\rVert_1 \label{eq:pbc-sound}
    \end{equation}
    where $H_{\mathbb{Z}^2, \text{pbc}}$ denotes the Hamiltonian lifted to all of $\mathbb{Z}^2$, but maintaining periodic boundary conditions $x_L=x_0$.
    Note there is no constant offset term here because of the periodic boundary conditions.  
    
     Concretely, we now show that large excursions in $\mathrm{tr}(K_\beta^L)$ to points where $\lVert x_i\rVert_1 \sim r_L$ contribute negligibly to the overall sum, which will eventually lead us to the proof of measure concentration.  
    Notice that, analogously to \eqref{eq:sigma_inequality}, \begin{align}
        \sum_{x\in \Omega \; : \;  \max  \lVert x_i\rVert_1\ge R} \mathrm{e}^{-\beta H_{\mathbb{Z}^2,\text{pbc}}(x)} &< \sum_{x\in \Omega \; : \;  \max  \lVert x_i\rVert_1\ge R}  \prod_{i=0}^{L-1}\mathrm{e}^{-\beta \eta \lVert x_i\rVert_1}  \notag \\
        &<L\kappa_{\beta \eta}^{L-1}\sum_{x_0 : \lVert x_0\rVert_1\ge R}\mathrm{e}^{-\beta \eta \lVert x_0\rVert_1} 
        < L\kappa_{\beta \eta}^{L-1}\frac{4R\mathrm{e}^{-\beta \eta R}}{(1-\mathrm{e}^{-\beta \eta})^2}.
    \end{align}
    For $\gamma$ defined in Proposition \ref{prop:GL-middle}, take \begin{equation}\label{eq:define_R}
        R = \frac{1}{3}\gamma r_L = \mathrm{e}^{\mathrm{\Theta}(L)}.
    \end{equation}  For large enough $L$, the double exponential scaling in $L$ will overwhelm everything else, and we conclude that the probability of a large excursion in the lifted problem on $\mathbb{Z}^2$ is negligibly small: for some $L$-independent constant $0<C_1<\infty$, \begin{equation}
        \mathbb{P}_{\mathbb{Z}^2,\text{pbc}}\left[ \max  \lVert x_i\rVert_1\ge R \right] = \frac{1}{\mathrm{tr}(K_\beta^L)}\sum_{x\in \Omega:  \sum \lVert x_i\rVert_1\ge R} \mathrm{e}^{-\beta H_{\mathbb{Z}^2,\text{pbc}}(x)}  < C_1\mathrm{e}^{-\beta \eta R/2}. \label{eq:PZ2pbc}
    \end{equation}
    We used a crude estimate, as in \eqref{eq:KLbeta1}, to show that the overall partition function cannot be too small.

  Consider the subspace of state space \begin{equation}
        \Omega_0 := \left\lbrace x\in (G_L)^L :  \max_i \lVert x_i\rVert_{\mathrm{L}} < R\right\rbrace .
    \end{equation}
    Define $\Omega_a = U_a \cdot \Omega_0$.  Proposition \ref{prop:GL-middle} shows that $x\in \Omega_0$ is its own unique smallest 1-norm lift in $\mathbb{Z}^2$, and that for any two $x_i, x_{i+1}\in \Omega_0$, \begin{equation}
        \lVert x_i - Ax_{i+1}\rVert_{\mathrm{L}} \le \left(1+\lVert A\rVert_{1\rightarrow 1}\right) R = \frac{r_L}{6}.
    \end{equation}
    In contrast, the lifted orbits cannot be so small:
    \begin{lem}
        If $a\ne a^\prime$, $\Omega_a \cap \Omega_{a^\prime}=\emptyset$.  
    \end{lem}
    \begin{proof}
   It suffices to prove this for the case $a^\prime=0$, since if $\Omega_{a-a^\prime} \cap \Omega_0=\emptyset$ (with $a\ne a^\prime$), then $\Omega_a \cap \Omega_{a^\prime} = \emptyset$ too, since $U_a$ is an invertible group action.

 \begin{lem}
     For every $a\in G_L\setminus\lbrace 0\rbrace$,  for at least one $j < L$, \begin{equation}
      \left\lVert A^{-j} a \right\rVert_{\mathrm{L}} \ge \gamma r_L, \label{eq:Aja_big}
  \end{equation}
 \end{lem}
 \begin{proof}
  We prove by contradiction.  Let $v_j \in A^{-j}a+M_L$ be the smallest Lee weight representative of each member of the orbit.  If \eqref{eq:Aja_big} is false, then \begin{equation}
      \lVert v_j- A v_{j+1} \rVert_1 \le  \left(1+\lVert A\rVert_{1\rightarrow 1}\right)\gamma r_L = \frac{r_L}{2}. \label{eq:rL2}
  \end{equation}
Since $\lVert v_j - Av_{j+1}\rVert_{\mathrm{L}}=0$, namely $v_j - Av_{j+1}\in M_L$, combining \eqref{eq:rL2} with Proposition \ref{prop:GL-middle} implies $v_j = Av_{j+1}$.  Hence $v_0 = A^Lv_L = A^Lv_0$ which forces $v_0=0$.  But no $a\ne 0$ lies in the same coset as 0, which is the desired contradiction.
 \end{proof}

  Fix $a$, and let $j$ be one index where \eqref{eq:Aja_big} holds.  Take any $x\in \Omega_0$ and $y\in \Omega_a$.  Observe that \begin{align} \label{eq:yj-different}
      \lVert y_j \rVert_{\mathrm{L}} &= \lVert (y_j- A^{-j}a) +  A^{-j}a \rVert_{\mathrm{L}} \ge \lVert A^{-j}a\rVert_{\mathrm{L}} - \lVert y_j - A^{-j}a\rVert_{\mathrm{L}} > \gamma r_L \left(1 - \frac{1}{3}\right) = 2R. 
  \end{align}
  This ensures that $\Omega_0\cap \Omega_a=\emptyset$.
    \end{proof}

A similar argument to \eqref{eq:yj-different} shows that \begin{align}
      \lVert x_j - y_j \rVert_{\mathrm{L}} &= \lVert x_j - (y_j- A^{-j}a) -  A^{-j}a \rVert_{\mathrm{L}} \ge \lVert A^{-j}a\rVert_{\mathrm{L}} - \lVert x_j\rVert_{\mathrm{L}} - \lVert y_j - A^{-j}a\rVert_{\mathrm{L}} \notag \\
      &> \gamma r_L \left(1 - \frac{1}{3}-\frac{1}{3}\right) = R. 
  \end{align}
  This establishes \eqref{eq:cluster-distance-big}.

  To complete the proof we derive \eqref{eq:mu_Omega_a}. Notice that $\langle x|\mathcal{K}_\beta|y\rangle = \langle x|K_\beta|y\rangle$ if $x,y\in \Omega_0$.  Combining Lemmas \ref{lem:Kbeta} and \ref{lem:lift-to-Z2} with \eqref{eq:PZ2pbc}, \begin{align}
        \mu_\beta(\Omega_0) &= \frac{1}{Z_\beta}\sum_{x_0,\ldots, x_{L-1}\in\Omega_0} \mathrm{e}^{-\beta H} \notag \\
        &\ge  \frac{1}{|G_L| \mathrm{tr}(K_\beta^L) (1+C_0 \mathrm{e}^{-\beta r_L/4})}\sum_{x_0,\ldots, x_{L-1}\in\Omega_0} \mathrm{e}^{-\beta H}   \notag \\
        &\ge  \frac{1}{|G_L|} \frac{1 - C_1 \mathrm{e}^{-\beta \eta \gamma r_L/3}}{1+C_0 \mathrm{e}^{-\beta r_L/4}}
    \end{align}
    which establishes \eqref{eq:mu_Omega_a} for all $\Omega_a$ since $H$, after taking $C = \max(C_0,C_1)$ and $c=\max(\frac{1}{4},\frac{1}{3}\eta \gamma)$.  Hence the unique Gibbs measure, is invariant under the symmetry action of $U_a$.   We can then define the ``bottleneck" set \begin{equation}
        \Omega_{\text{bn}} := \Omega \setminus \bigcup_{a\in G_L}\Omega_a,
    \end{equation}
    and observe that $\mu(\Omega_{\text{bn}}) = \exp(-\beta \exp[\mathrm{\Theta}(L)])$.
\end{proof}

Analogous arguments to the proof of Corollary \ref{cor:deviation} imply that the Gibbs measure $\mu_\beta$ conditioned on $\Omega_a$ is very simple:  each coordinate $x_j \approx A^{-j}a \text{ (mod $M_L$)}$, with large excursions exponentially suppressed; we omit them for brevity.  

The crucial feature of Theorem \ref{thm:finite} is that it provides an independent justification for the claim that the DLR states $\mu^c_\beta$ should be called \textbf{spontaneous symmetry breaking} states, as opposed to explicit symmetry breaking.  The measure concentration phenomenon in the theorem above into disjoint and far-separated clusters related by symmetry is exactly what happens in conventional lattice models, such as the two-dimensional Ising model \cite{Thomas1989IsingMassGap}, in any finite but large system size.  

A dynamical interpretation of the clustering phenomenon \eqref{eq:mu_Omega_a} is as follows: dynamics which changes states by a finite Lee-weight move in each step has an exponentially slow mixing time \cite{leecode1}; this result follows immediately from standard Markov chain mixing time bounds \cite{Levin:2017mixing}.  

\section{Remarks on a continuous helical symmetry}\label{sec:continuous}
In this section, we will briefly comment on the following extension of our model to a continuous state space.   We will again keep the discussion non-rigorous; we expect that similar techniques to the ones we developed above can also be used to analyze this model rigorously, but with mathematical complications due to the continuous state space.  The model is also more analogous to earlier models, e.g. \cite{Niedermaier:2003cq}.

Consider state space $\Omega = (\mathbb{R}^2)^{\mathbb{Z}}$, and pick any hyperbolic symmetric matrix $A\in \mathrm{SL}(2,\mathbb{R})$.   Define the Hamiltonian \begin{equation}
    H = \sum_{n\in \mathbb{Z}} \lVert x_n - Ax_{n+1}\rVert_1,
\end{equation}
with analogous caveats to before: formally $H$ is not-defined, and one must truncate it to finite-sized regions, as usual.  This model now has an enhanced helical symmetry $\mathbb{R}^2\rtimes \mathbb{Z}$ analogous to \eqref{eq:helical-group-action}.

It is easy to generalize Theorem \ref{thm:sound} for this model, so we expect that an analogue to Theorem \ref{thm:exist} is possible, along with an analogue to Corollary \ref{cor:deviation}; however, one needs to work with a more general theory of Markov chains on continuous spaces.\footnote{This may complicate the proof of Theorem \ref{thm:exist} in particular due to the critical role that path coupling on a discrete case played in that proof.}  One interesting physical remark is that, unlike in the case of the discrete helical symmetry $\mathbb{Z}^2\rtimes \mathbb{Z}$, in the helical case, all DLR states that can be constructed analogous to the ones in this paper would be related to each other by symmetry. 

There is a particularly nice physical interpretation of the generalization of Theorem \ref{thm:exist} in this context.  Recently, \cite{lian3} argued that in a model with an ``exponential translationally covariant modulated symmetry" (a.k.a. our helical symmetry), the Goldstone modes for the resulting symmetry are massive.   The ``mass" of the Goldstone is bounded by a soundness guarantee from the natural generalization of Theorem \ref{thm:exist}.   Since the proof of the standard Mermin-Wagner Theorem forbidding spontaneous symmetry breaking in $d\le 2$ spatial dimensions is crucially tied to the massless nature of the Goldstone bosons \cite{Mermin:1966absence,Hohenberg:1967existence}, this provides a more intuitive physical understanding for why symmetry breaking is possible.  Although the breakdown of the Mermin-Wagner Theorem in a setting where the broken symmetry mixes with translation is not too unexpected \cite{Low:2002spacetime,Halperin:2019limitations}, we should stress that multipolar symmetries can mix with translation but also have Mermin-Wagner-type bounds \cite{Stahl:2021sgi,Kapustin:2022fzp}, although there are major differences between lattice and continuum models \cite{Glorioso:2023chm,Jain:2023nbf} with regards to spontaneously broken multipolar symmetries.

\section{Outlook}
We have constructed an uncountable infinity of DLR states for a one-dimensional model with on-site $\mathbb{Z}^2$ state space.   Any one such DLR state is related to countably infinitely many others by a free action of a $\mathbb{Z}^2$ helical flavor symmetry group.  As such, we interpret this model as exhibiting spontaneous symmetry breaking at any temperature.   This is a counterexample to the standard lore that such symmetry breaking is not possible in one dimension, although the infinite state space exploits an obvious loophole in standard theorems \cite{barrysimon}.    In contrast to earlier constructions such as \cite{Kittel:1969zipper,benfatto,Niedermaier:2003cq,Spitzer:1975phase}, interpreting the existence of multiple infinite-volume Gibbs states as ordinary symmetry breaking is more transparent.

It would be fascinating if it is possible to use infinite state spaces to spontaneously break flavor symmetries, that are not helical (linked with translational symmetry), in dimensions in violation of a Peierls argument or Mermin-Wagner Theorem.  At first glance, it seems impossible to do this, especially if -- as in the models in the present paper, and in \cite{Han:2025eiw} -- the Gibbs state marginalized on a single site is normalizable even in the thermodynamic limit.   But we are also not aware of any formal theorem that enforces a particular helical symmetry structure whenever a continuous symmetry is spontaneously broken below three dimensions;  establishing one could be an interesting goal for future work.

It seems straightforward to extend this model to a many-body quantum system. One way to do this is to take the regularized model from Section \ref{sec:finite-size}, and to write \begin{equation}
    H_{\text{quantum}} = \sum_{x\in \Omega} H(x)|x\rangle\langle x| - \Gamma \sum_{i\in \mathbb{Z}_L}\sum_{x\in \Omega} \sum_{y:\lVert y\rVert_1 = 1}|x_{-i}, x_i+y\rangle\langle x_{-i},x_i |. \label{eq:H_quantum}
\end{equation}
In other words, the transverse quantum fluctuations can only flip a single ``spin" at a time, and only by $\pm 1$ in one of the two directions on the on-site coset space $G_L$.   One may also be able to study such a Hamiltonian on the original infinite-dimensional line, using the same language of $C^*$ algebras that can be applied to quantum spin chains on the infinite line \cite{Nachtergaele:2009harmonic,Nachtergaele:2014unbounded}.  

Now return to the model on a finite-dimensional Hilbert space, where we use the regularized model on a chain of length $L$ with periodic boundary conditions.  We expect that, following \cite{Yin:2024jad}, it is straightforward to show that all eigenstates whose energy is low compared to the free-energetic bottleneck separating symmetry sectors are (many-body) Anderson localized, so long as $H_{\text{quantum}}$ is perturbed by some randomness that explicitly breaks the helical symmetry.  The strength of this randomness needed for localization will likely be extremely small. This would be a pathological construction of many-body localization in a local one-dimensional quantum system with an enormous on-site Hilbert space.  It may still be interesting because establishing many-body localization at finite energy density in a one-dimensional spin chain of interacting qubits is, in contrast, a notoriously difficult problem in mathematical physics \cite{Imbrie:2016mbl,deroeck2025absencenormalheatconduction,elgart2026manybodylocalizationrandomxxz}.

 \section*{Acknowledgments}
\addcontentsline{toc}{section}{Acknowledgements}
The conceptual development of this project and the calculations were collaboratively performed with ChatGPT 5.6 Pro.  I wrote the final manuscript without AI assistance and take responsibility for the content.

\begin{appendix}

\makeatletter


\let\@sectioncntformat\@seccntformat

\let\@hangfrom@section\@hang@from


\renewcommand\section{\@ifstar{\AndySectionStar}{\@startsection{section}{1}{\z@}%

  {-3.5ex \@plus -1ex \@minus -.2ex}%

  {2.3ex \@plus .2ex}%

  {\normalfont\Large\bfseries}}}

\renewcommand\subsection{\@ifstar{\AndySubsectionStar}{\@startsection{subsection}{2}{\z@}%

  {-3.25ex\@plus -1ex \@minus -.2ex}%

  {1.5ex \@plus .2ex}%

  {\normalfont\large\bfseries}}}

\renewcommand\subsubsection{\@ifstar{\AndySubsubsectionStar}{\@startsection{subsubsection}{3}{\z@}%

  {-3.25ex\@plus -1ex \@minus -.2ex}%

  {1.5ex \@plus .2ex}%

  {\normalfont\normalsize\bfseries}}}

\makeatother

\renewcommand{\thesubsection}{\thesection.\arabic{subsection}}
\renewcommand{\theequation}{\thesection.\arabic{equation}}
\renewcommand{\thesubsubsection}{\thesubsection.\arabic{subsubsection}}
\makeatletter
\renewcommand\section{\@ifstar{\AndySectionStar}{\@startsection {section}{1}{-\parindent}%
  {-3.5ex \@plus -1ex \@minus -.2ex}%
  {2.3ex \@plus .2ex}%
  {\normalfont\Large\bfseries}}}
\renewcommand\subsection{\@ifstar{\AndySubsectionStar}{\@startsection{subsection}{2}{-\parindent}%
  {-3.25ex\@plus -1ex \@minus -.2ex}%
  {1.5ex \@plus .2ex}%
  {\normalfont\large\bfseries}}}
\renewcommand\subsubsection{\@ifstar{\AndySubsubsectionStar}{\@startsection{subsubsection}{3}{-\parindent}%
  {-3.25ex\@plus -1ex \@minus -.2ex}%
  {1.5ex \@plus .2ex}%
  {\normalfont\normalsize\bfseries}}}
\makeatother

\end{appendix}

\phantomsection
\section*{References}
\addcontentsline{toc}{section}{References}
\makeatletter
\let\bibsection\relax
\makeatother
\bibliography{thebib}
\end{document}